\documentclass[10pt,twocolumn,letterpaper]{article}

\usepackage{cvpr}              % To produce the CAMERA-READY version
\usepackage{graphicx}
\usepackage{amsmath}
\usepackage{amssymb}
\usepackage{booktabs}

\usepackage{epsfig}
\usepackage{cite}
\usepackage{verbatim}
\usepackage{color}
\usepackage{array}
\usepackage{epstopdf}
\usepackage{booktabs, multirow}       % professional-quality tables
\usepackage{stackengine}
\usepackage{colortbl}
\usepackage{multirow}
\usepackage{makecell}
\usepackage{pifont}

\usepackage{amsthm}
\usepackage[accsupp]{axessibility}  % Improves PDF readability for those with disabilities.

\newtheorem{theorem}{Theorem}

\newcommand{\bX}{\ensuremath{{\mathbf{X}}}}
\newcommand{\bXh}{\ensuremath{{\widehat{\mathbf{X}}}}}

\newcommand{\bY}{\ensuremath{{\mathbf{Y}}}}
\newcommand{\bYh}{\ensuremath{{\widehat{\mathbf{Y}}}}}
\newcommand{\bYt}{\ensuremath{{\widetilde{\mathbf{Y}}}}}

\newcommand{\bZ}{\ensuremath{{\mathbf{Z}}}}
\newcommand{\bZh}{\ensuremath{{\widehat{\mathbf{Z}}}}}

\newcommand{\bM}{\ensuremath{{\mathbf{M}}}}
\newcommand{\bSig}{\ensuremath{{\mathbf{\Sigma}}}}

\newcommand{\bT}{\ensuremath{{\mathbf{T}}}}

\newcommand{\bE}{\ensuremath{{\mathbf{E}}}}
\newcommand{\bB}{\ensuremath{{\mathbf{B}}}}

\newcommand{\bDel}{\ensuremath{{\mathbf{\Delta}}}}
\newcommand{\bP}{\ensuremath{{\mathbf{P}}}}

\newcommand{\bPt}{\ensuremath{{\widetilde{\mathbf{P}}}}}

\newcommand{\bS}{\ensuremath{\mathbf{S}}}

\newcommand{\RNum}[1]{\uppercase\expandafter{\romannumeral #1\relax}}

\newcommand{\Appenproof}{\color{red}A\color{black}}
\newcommand{\Appentrainingdetail}{\color{red}B\color{black}}
\newcommand{\Appenresults}{\color{red}C\color{black}}

\def\@fnsymbol#1{\ensuremath{\ifcase#1\or \dagger\or \ddagger\or
   \mathsection\or \mathparagraph\or \|\or **\or \dagger\dagger
   \or \ddagger\ddagger \else\@ctrerr\fi}}

\usepackage[pagebackref=true,breaklinks=true,letterpaper=true,colorlinks,bookmarks=false]{hyperref}

\usepackage[capitalize]{cleveref}
\crefname{section}{Sec.}{Secs.}
\Crefname{section}{Section}{Sections}
\Crefname{table}{Table}{Tables}
\crefname{table}{Tab.}{Tabs.}

\def\cvprPaperID{6400} % *** Enter the CVPR Paper ID here
\def\confName{CVPR}
\def\confYear{2023}

\begin{document}

\title{Context-Based Trit-Plane Coding for Progressive Image Compression}

\author{Seungmin Jeon$^1$, Kwang Pyo Choi$^2$, Youngo Park$^2$, Chang-Su Kim$^{1}\thanks{Corresponding author.}$\\
$^1$Korea University, $^2$Samsung Electronics\\
{\tt\footnotesize seungminjeon@mcl.korea.ac.kr,
\tt\footnotesize \{kp5.choi, youngo.park\}@samsung.com,
\tt\footnotesize changsukim@korea.ac.kr}
}
\maketitle

%%%%%%%%% ABSTRACT
\begin{abstract}
Trit-plane coding enables deep progressive image compression, but it cannot use autoregressive context models. In this paper, we propose the context-based trit-plane coding (CTC) algorithm to achieve progressive compression more compactly. First, we develop the context-based rate reduction module to estimate trit probabilities of latent elements accurately and thus encode the trit-planes compactly. Second, we develop the context-based distortion reduction module to refine partial latent tensors from the trit-planes and improve the reconstructed image quality. Third, we propose a retraining scheme for the decoder to attain better rate-distortion tradeoffs. Extensive experiments show that CTC outperforms the baseline trit-plane codec significantly, \textit{e.g.}~by $-14.84\%$ in BD-rate on the Kodak lossless dataset, while increasing the time complexity only marginally. The source codes are available at \href{https://github.com/seungminjeon-github/CTC}{https://github.com/seungminjeon-github/CTC}.
\end{abstract}
\vspace*{-0.4cm}

%%%%%%%%% BODY TEXT
\section{Introduction}
\label{sec:intro}
Image compression is a fundamental problem in both image processing and low-level vision. A lot of traditional codecs have been developed, including standards  JPEG~\cite{y1992_CE_JPEG}, JPEG2000 \cite{y2001_SPM_JPEG2000}, and VVC \cite{y2021_TCSVT_Bross}. Many of these codecs are based on discrete cosine transform or wavelet transform. Using handcrafted modules, they provide decent rate-distortion (RD) results. However, with the rapidly growing usage of image data, it is still necessary to develop advanced image codecs with better RD performance.

Deep learning has been explored with the advance of big data analysis and computational power, and it also has been successfully adopted for image compression. Learning-based codecs have similar structures to traditional ones: they transform an image into latent variables and then encode those variables into a bitstream. They often adopt convolutional neural networks (CNNs) for the transformation. Several innovations have been made to improve RD performance, including differentiable quantization approximations \cite{y2017_ICLR_balle,y2017_NIPS_agustsson}, hyperprior \cite{y2018_ICLR_balle}, context models \cite{y2018_NIPS_minnen,y2018_CVPR_mentzer,y2021_CVPR_he}, and prior models \cite{y2020_CVPR_cheng,y2021_CVPR_cui}. As a result, the deep image codecs are competitive with or even superior to the traditional ones.

\begin{figure}[t]
    \begin{center}
    \includegraphics[width=\linewidth]{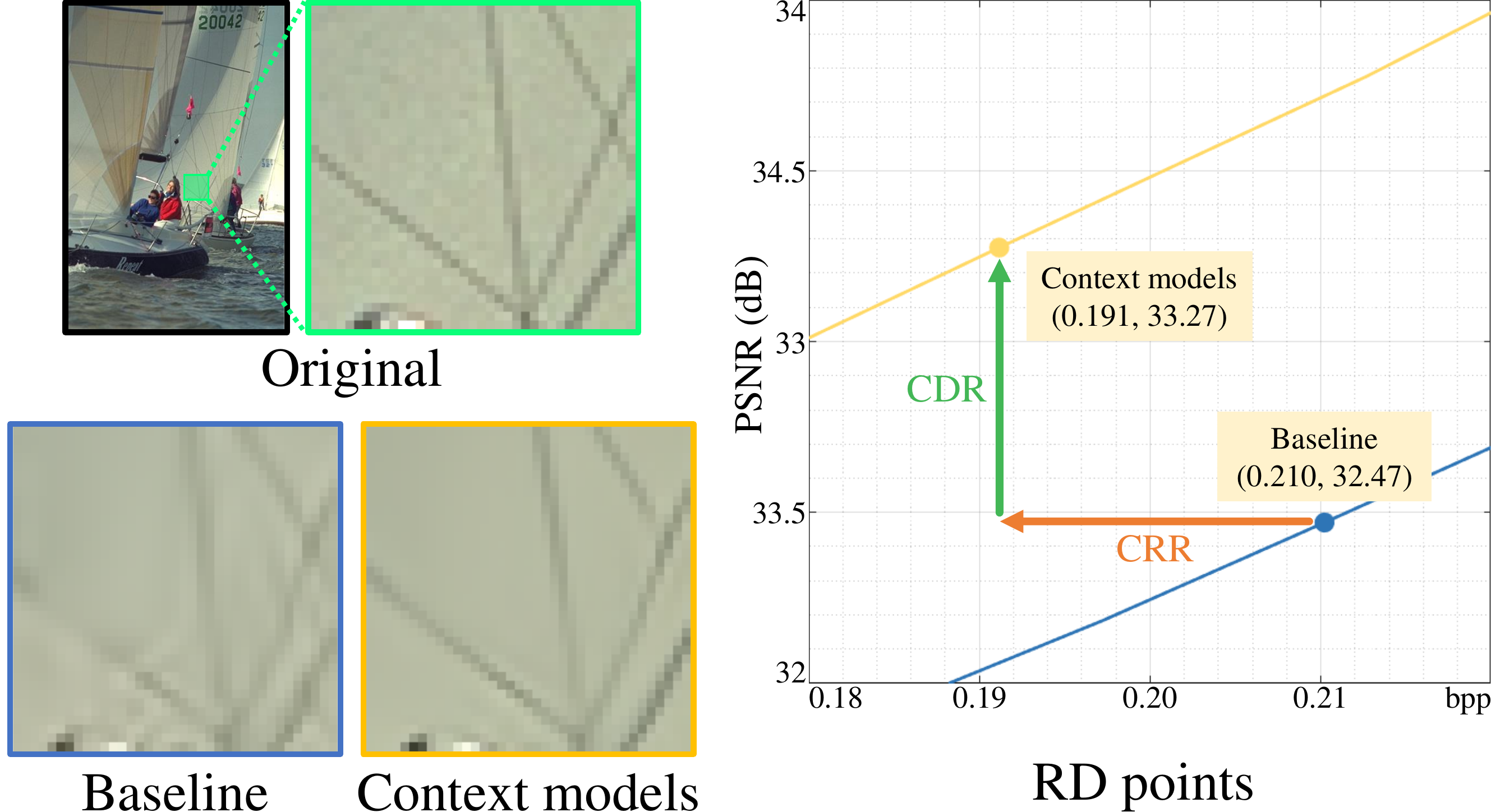}
    \end{center}
    \vspace*{-0.5cm}
    \caption
    {
       Illustration of the proposed context models: CRR reduces the bitrate, while CDR improves the image quality, as compared with the context-free baseline \cite{y2022_CVPR_lee}.
    }
    \label{fig:intro}
\end{figure}

It is desirable to compress images progressively in applications where a single bitstream should be used for multiple users with different bandwidths. But, relatively few deep codecs support such progressive compression or scalable coding \cite{y2005_IEEE_ohm}. Many codecs should train their networks multiple times to achieve compression at as many bitrates \cite{y2018_ICLR_balle,y2018_NIPS_minnen,y2020_CVPR_cheng,y2022_ICLR_zhu}. Some codecs support variable-rate coding \cite{y2021_CVPR_cui,y2021_CVPR_yang}, but they should generate multiple bitstreams for different bitrates. It is more efficient to truncate a single bitstream to satisfy different bitrate requirements. Lu \etal \cite{y2021_ICIP_lu} and Lee \etal \cite{y2022_CVPR_lee} are such progressive codecs, based on nested quantization and trit-plane coding, respectively. But, they cannot use existing context models \cite{y2018_NIPS_minnen,y2018_CVPR_mentzer,y2019_ICLR_lee,y2020_ICIP_minnen,y2021_CVPR_he}, which assume the synchronization of the latent elements, used as contexts, in the encoder and the decoder. Those latent elements are at different states depending on bitrates.

In this paper, we propose the context-based trit-plane coding (CTC) algorithm for progressive image compression, based on novel context models. First, we develop the context-based rate reduction (CRR) module, which entropy-encodes trit-planes more compactly by exploiting already decoded information. Second, we develop the context-based distortion reduction (CDR) module, which refines partial latent tensors after entropy decoding for higher-quality image reconstruction. Also, we propose a simple yet effective retraining scheme for the decoder to achieve better RD tradeoffs. It is demonstrated that CTC outperforms the existing progressive codecs \cite{y2021_ICIP_lu,y2022_CVPR_lee} significantly.

This paper has the following major contributions:
\begin{itemize}
\itemsep0em
\item We propose the \textit{first} context models, CRR and CDR, for deep progressive image compression. As illustrated in Figure \ref{fig:intro}, CRR reduces the bitrate, while CDR improves the image quality effectively, in  comparison with the baseline trit-plane coding \cite{y2022_CVPR_lee}.
\item We develop a decoder retraining scheme, which adapts the decoder to refined latent tensors by CDR to improve the RD performance greatly.
\item The proposed CTC algorithm outperforms the state-of-the-art progressive codecs \cite{y2021_ICIP_lu, y2022_CVPR_lee} significantly. Relative to     \cite{y2022_CVPR_lee}, CTC yields BD-rates of $-14.84\%$ on the Kodak dataset \cite{kodim}, $-14.75\%$ on the CLIC validation set \cite{clic2021}, and $-17.00\%$ on the JPEG-AI testset \cite{jpegai}.
\end{itemize}

%------------------------------------------------------------------------
\section{Related Work}
\label{sec:related}

\noindent
\textbf{Learning-based codecs:} Early learning-based image codecs \cite{y2016_NIPS_gregor,y2016_ICLR_toderici,y2017_CVPR_toderici,y2018_CVPR_johnston} are based on recurrent neural networks (RNNs), but more codecs \cite{y2017_ICLR_balle,y2017_ICLR_theis,y2018_ICLR_balle,y2018_NIPS_minnen,y2020_CVPR_cheng} employ CNN-based autoencoders \cite{y2008_ICML_vincent}. Ball{\'e} \etal \cite{y2017_ICLR_balle} proposed an additive noise model to approximate quantization and trained their network in an end-to-end manner. In \cite{y2018_ICLR_balle,y2018_NIPS_minnen}, hyperprior information was used to estimate the probability distributions of latent elements more accurately. Cheng \etal \cite{y2020_CVPR_cheng} used residual blocks and attention modules in the autoencoder and adopted a Gaussian mixture prior.

Recently, vision transformer \cite{y2021_ICLR_dosovitskiy} or self-attention has been adopted to yield better RD results \cite{y2022_ICLR_zhu,y2022_ICLR_qian,y2022_CVPR_zou,y2022_CVPR_kim}. Qian \etal \cite{y2022_ICLR_qian} developed transformer-based hyper-encoder and hyper-decoder. Kim \etal \cite{y2022_CVPR_kim} decomposed hyperprior parameters to global and local ones. Zou \etal \cite{y2022_CVPR_zou} used window attention modules in their CNN-based encoder and decoder. Zhu \etal \cite{y2022_ICLR_zhu} adopted the Swin transformer\cite{y2021_ICCV_liu} for their encoder, decoder, hyper-encoder and hyper-decoder.

\vspace{0.1cm}
\noindent
\textbf{Variable-rate compression:} The aforementioned codecs can compress an image at a single rate only. For variable-rate compression, they should be trained multiple times, which is inefficient in both time and memory. In contrast, there are several variable-rate codecs  \cite{y2017_ICLR_theis,y2019_ICCV_choi,y2021_CVPR_yang,y2021_CVPR_cui,y2021_ICCV_song}. Theis \etal \cite{y2017_ICLR_theis} and Choi \etal \cite{y2019_ICCV_choi} adopted scale parameters for quantization to achieve variable-rate coding. Yang \etal \cite{y2021_CVPR_yang} adopted the slimmable neural networks\cite{y2018_ICLR_yu} and used subsets of network parameters to control bitrates. Cui \etal \cite{y2021_CVPR_cui} proposed the gain unit for channel-wise bit allocation. Song \etal \cite{y2021_ICCV_song} utilized a pixelwise quality map for rate control. These  variable-rate codecs support multiple bitrates via single network training, but they still generate separate bitstreams at different bitrates.

\vspace{0.1cm}
\noindent
\textbf{Progressive compression:} A single bitstream can support multiple bitrates in progressive compression. For example, the traditional JPEG and JPEG2000 have optional progressive modes \cite{y1992_CE_JPEG,y2001_SPM_JPEG2000}. Most of learning-based progressive codecs are based on RNNs \cite{y2016_NIPS_gregor,y2016_ICLR_toderici,y2017_CVPR_toderici,y2018_CVPR_johnston}, which support a limited number of quality levels. Also, Cai \etal \cite{y2019_PCS_cai} supports only two quality levels with two decoders.

It is more desirable to offer fine granular scalability (FGS) \cite{y1996_TCSVT_said,y2001_TCSVT_li}: a single bitstream can be truncated at any point for the decoder to reconstruct an image. Lu \etal \cite{y2021_ICIP_lu} used nested quantization for FGS. Lee \etal \cite{y2022_CVPR_lee} proposed trit-plane coding and RD-prioritized transmission of trits. These FGS codecs yield comparable RD curves to conventional deep image codecs.

\vspace{0.1cm}
\noindent
\textbf{Context models:} As context-based entropy coding techniques such as CABAC \cite{y2003_TCSVT_marpe} are used in traditional codecs \cite{TextComp,y2003_TCSVT_wiegand}, context models are also employed in learning-based codecs \cite{y2018_NIPS_minnen,y2018_CVPR_mentzer,y2019_ICLR_lee,y2020_ICIP_minnen,y2021_CVPR_he}. Minnen \etal\cite{y2018_NIPS_minnen} and Mentzer \etal\cite{y2018_CVPR_mentzer} developed autoregressive context models using masked 2D and 3D CNNs, respectively. The autoregressive models exploit spatial contexts serially, demanding high time complexity. Lee \etal \cite{y2019_ICLR_lee} proposed bit-consuming and bit-free contexts to estimate latent distributions. Minnen \etal \cite{y2020_ICIP_minnen} explored a channelwise autoregressive model with latent residual prediction. He \etal \cite{y2021_CVPR_he} developed a checkerboard context model to reduce time complexity.

All these context models can be used for fixed-rate compression only. In contrast, we develop two context models, CRR and CDR, for progressive compression based on trit-plane coding, which improve the RD performance significantly with only a marginal increase of time complexity.

\begin{figure*}[h]
    \begin{center}
    \includegraphics[width=\linewidth]{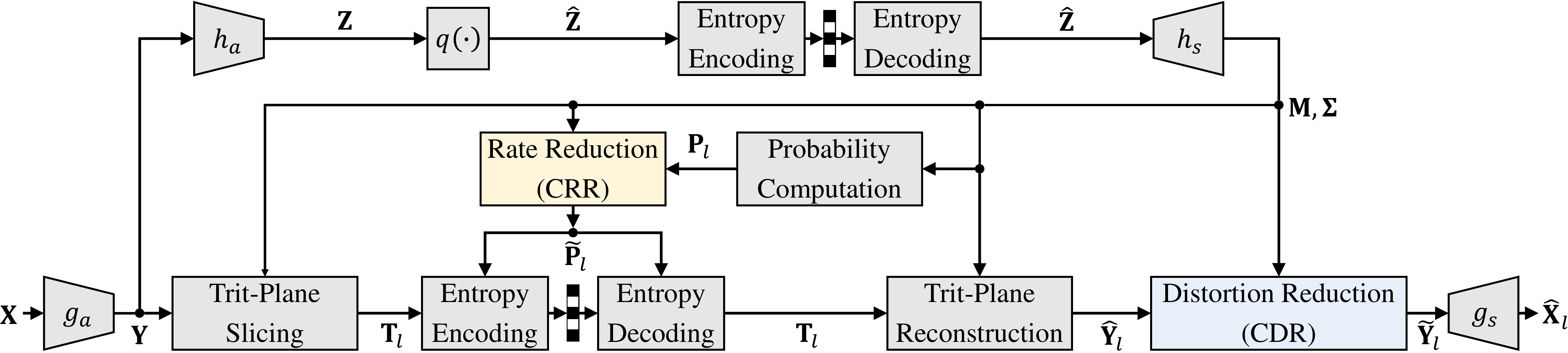}
    \end{center}
    \vspace*{-0.3cm}
    \caption
    {
        The framework of the proposed CTC algorithm. The context-based rate reduction (CRR) and context-based distortion reduction (CDR) modules are shown in detail in Figure~\ref{fig:context_based_modules}.
    }
    \label{fig:model_architecture}
\end{figure*}

\begin{figure}[t]
    \begin{center}
    \includegraphics[width=\linewidth]{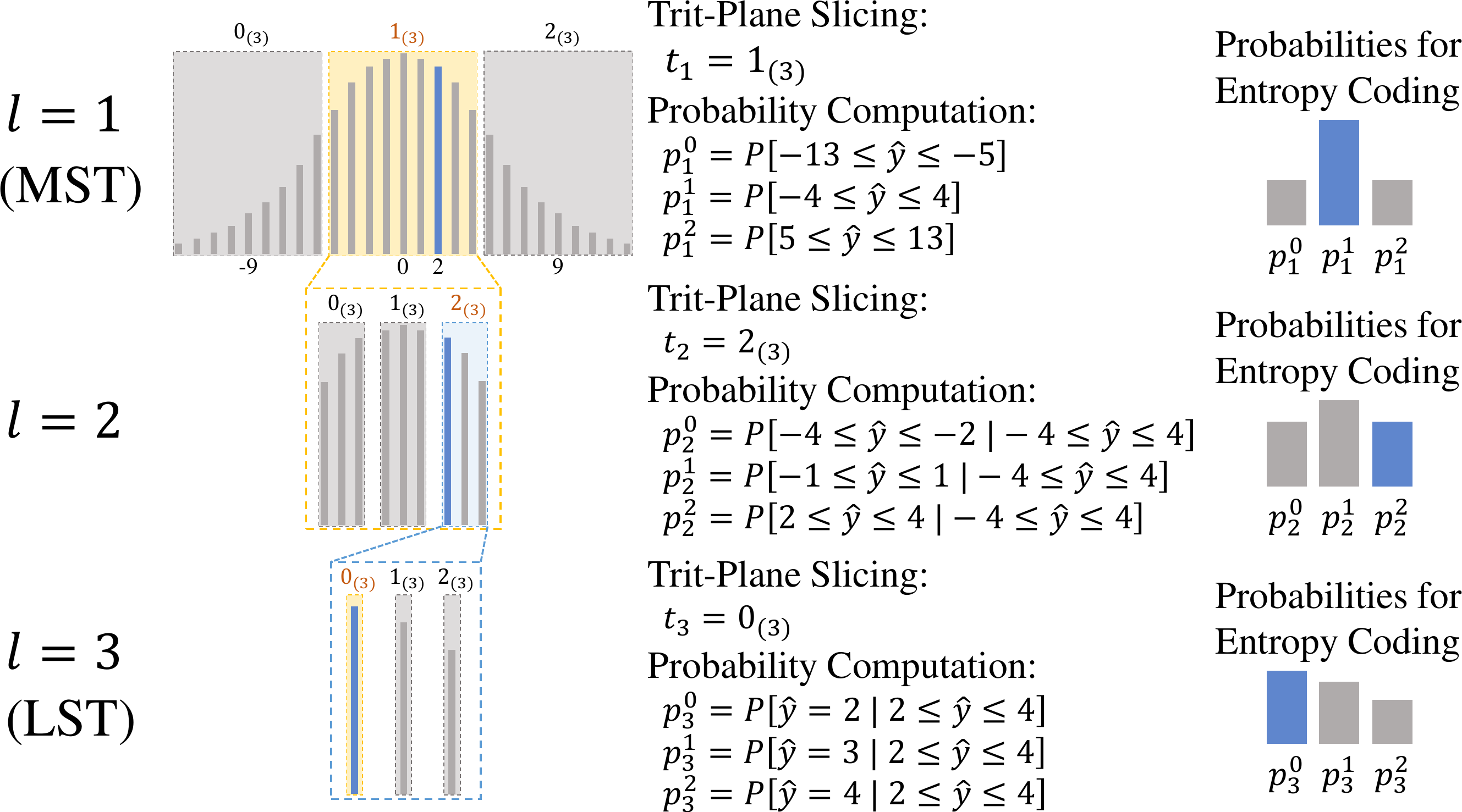}
    \end{center}
    \vspace*{-0.5cm}
    \caption
    {
       A toy example for trit-plane slicing and probability computation, where an element $\hat{y}$ in $\bYh$ equals $2$. In trit-plane slicing, we determine where the element belongs among three equal subintervals. In probability computation,  we compute the conditional probabilities of each trit. These two steps of trit-plane slicing and probability computation are carried out recursively from MST to LST.
    }
    \label{fig:toy_example}
\end{figure}

%------------------------------------------------------------------------
\section{Proposed Algorithm}
\label{sec:proposed}

\subsection{Trit-Plane Coding}
\label{ssec:backgrounds}

Trit-plane coding was introduced in \cite{y2022_CVPR_lee} for deep progressive image compression. Figure~\ref{fig:model_architecture} shows the framework of the proposed CTC algorithm, which is also based on the trit-plane representation of latent elements. The encoder $g_a$ and the hyper-encoder $h_a$ transform an image $\bX$ into a latent tensor $\bY$ and a hyper latent tensor $\bZ$ sequentially. Then, using the quantized $\bZh$, the hyper-decoder $h_s$ yields $\bM$ and $\bSig$, representing the mean and standard deviation of $\bY$.

For trit-plane coding, we express the centered and quantized latent tensor $\bYh=q(\bY-\bM)$ in a ternary number system through the trit-plane slicing module: $\bYh \in \mathbb{R}^{C \times H \times W}$ is sliced into $L$ trit-planes $\bT_{l}$, $l = 1, \ldots, L$. Each trit-plane is a tensor of the same size as $\bYh$. Also, $\bT_1$ is the most significant trit-plane (MST), while $\bT_L$ is the least significant one (LST). The trit-planes are entropy-encoded into a bitstream progressively from MST to LST. To entropy-encode the $l$th trit-plane $\bT_{l}$, we compute the probability tensor $\bP_l$ containing the probabilities that each trit in $\bT_{l}$ equals $0_{(3)}$, $1_{(3)}$, or $2_{(3)}$.\footnote{The subscripts $(3)$, indicating the ternary number system, are omitted in the remaining paper for notational convenience.} Thus, $\bP_l \in \mathbb{R}^{3C \times H \times W}$. In Figure~\ref{fig:model_architecture}, the probability computation module estimates $\mathbf{P}_l$ by employing the entropy parameters $\bM$, $\bSig$ and the already encoded trit-planes $\bT_{1:l-1}$. Before the entropy coding, CTC refines $\bP_l$ to $\bPt_l$ using the CRR module. Then, the trits in $\bT_{l}$ are encoded into a bitstream in the decreasing order of their RD priorities \cite{y2022_CVPR_lee}. Figure~\ref{fig:toy_example} is a toy example for trit-plane slicing and probability computation.

Conversely, at the decoder side, the trit-planes are entropy-decoded from the bitstream. At any point of the entropy decoding, the image can be reconstructed. Assume that only the first $l$ trit-planes are decoded. Here, $l$ can be a fractional number. For example, if $l=2.31$, $\bT_{1}$ and $\bT_{2}$ are fully decoded, while 31\% of trits in $\bT_{3}$ are decoded. Then, the trit-plane reconstruction module obtains the partial latent tensor $\bYh_{l}$ from the $l$ trit-planes. Specifically, let $y$ be a latent element. Using the available trits, the decoder first identifies the interval ${\cal I}$ where $y$ belongs and then reconstructs it to the conditional mean, given by
\begin{equation}
\hat{y}_l = E[y|y \in {\cal I}].
\label{eq:reconstruction_level}
\end{equation}
Finally, the CDR module reduces distortions in $\bYh_l$ to yield $\bYt_l$, and the decoder $g_s$ reconstructs the image $\bXh_l$ from the refined latent tensor $\bYt_l$.

\subsection{Context-Based Rate Reduction}
\label{ssec:encoding}

Context models are useful for compressing correlated signals efficiently \cite{TextComp}. In the learning-based codecs, an autoregressive context model \cite{y2018_NIPS_minnen} predicts the entropy parameters of a latent element using already encoded elements and it improves the RD performance significantly. However, it is impossible to use the autoregressive model for trit-plane coding. The model assumes that $C\times H \times W$ latent elements are coded in the same raster scan order by both the encoder and the decoder. Hence, when trit-planes are only partially reconstructed, the decoder cannot perform the same prediction as the encoder, so the decoding breaks down \cite{y2022_CVPR_lee}.

\begin{figure*}[h]
    \begin{center}
    \includegraphics[width=\linewidth]{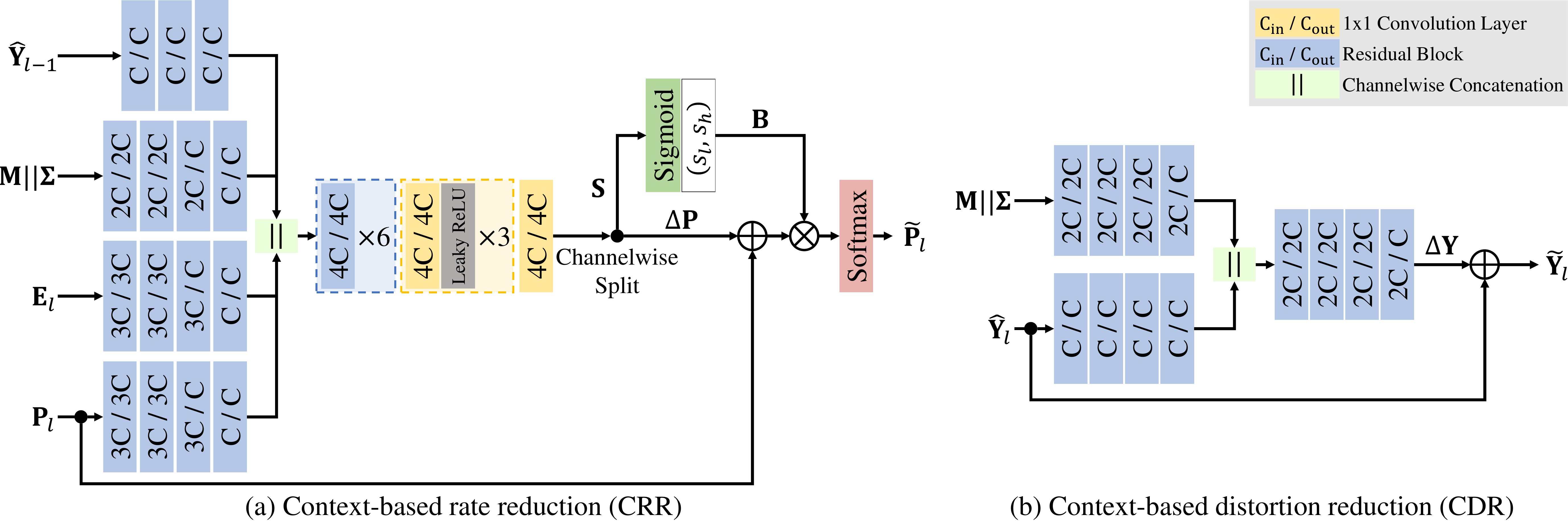}
    \end{center}
    \vspace*{-0.5cm}
    \caption
    {
        The architecture of the (a) CRR and (b) CDR modules. Each convolution layer has stride 1 and performs zero padding.
    }
    \vspace*{-0.1cm}
    \label{fig:context_based_modules}
\end{figure*}

\begin{figure}[t]
    \begin{center}
    \includegraphics[width=\linewidth]{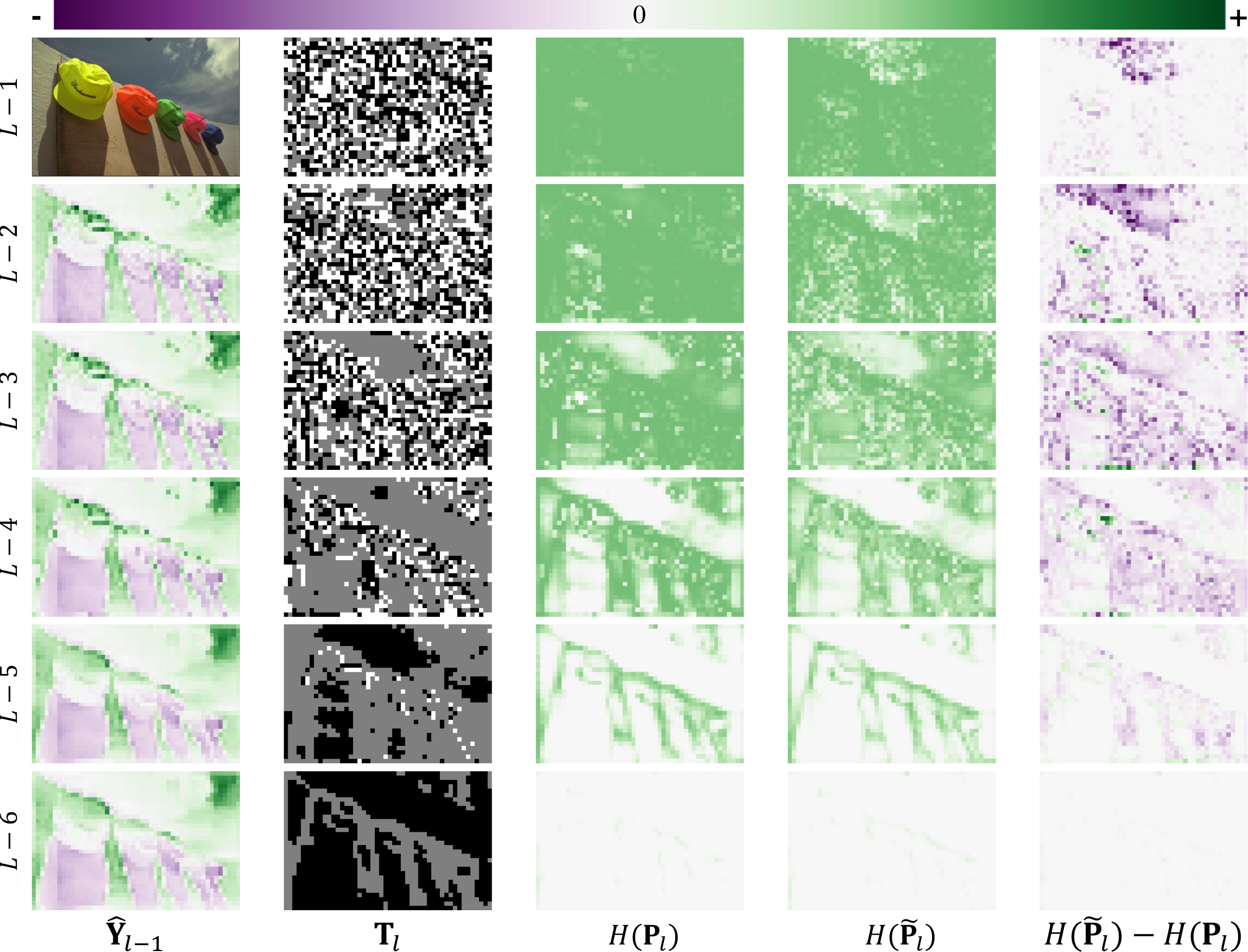}
    \end{center}
    \vspace*{-0.5cm}
    \caption
    {
        Visualization of $\bYh_{l-1}$, $\bT_l$, the entropies of $\bP_l$ and  $\bPt_l$, and the residual $H(\bPt_l)-H(\bP_l)$ in the 141st channel. The top left subfigure is, however, the original image for reference. In $\bT_l$, ternary values 0, 1, and 2 are shown in black, gray, and white. In the other cases, green and purple represent positive and negative values, as shown in the top color bar.
    }
    \label{fig:feature_visualization}
\end{figure}

We propose the first context models for trit-plane coding. Instead of predicting latent elements in the raster scan order, we predict each trit-plane $\bT_{l}$, $l=1, \ldots, L$, by exploiting already coded information, including the more significant trit-planes $\bT_{1:l-1}$. Note that the probability tensor $\bP_l$ is used to encode $\bT_l$. We refine the probability estimates in $\bP_l$ to yield an updated tensor $\bPt_{l}$ using the CRR module in Figure~\ref{fig:context_based_modules}(a). $\bPt_{l}$ requires fewer bits during the entropy coding than $\bP_l$ does, improving the RD performance.

To this end, we use already coded information: First, the approximate latent tensor $\bYh_{l-1}$, reconstructed from $\bT_{1:l-1}$, provides a context. Second, the entropy parameters $\bM$ and $\bSig$ are concatenated and used as another context. Third, the expected latent tensor $\bE_l$ is also used. Assuming that each trit in $\bT_{l}$ equals $0$, $1$, or $2$, the expected value of the corresponding latent element is computed via \eqref{eq:reconstruction_level}. $\bE_l$ contains the three possible values of every latent element, so $\bE_l \in \mathbb{R}^{3C \times H \times W}$.

In Figure \ref{fig:context_based_modules}(a), CRR extracts features from the input $\bP_l$ and the three contexts separately and fuses them lately through residual blocks and convolution layers. The fused tensor has the same spatial resolution as $\bP_l$ does, but four times more channels. It is split channelwise into an additive term $\bDel\!\bP \in \mathbb{R}^{3C \times H \times W}$ and a scaling term $\bS \in \mathbb{R}^{C \times H \times W}$. First, $\bS$ is converted into $\bB$ by
 \begin{equation}
\bB=s_l + (s_h - s_l) \times \textrm{sigmoid}(\bS),
\label{eq:beta_limit}
\end{equation}
whose each element is within $(s_l, s_h)$. Then, $\bP_l$ is added to $\bDel\!\bP$, and the sum is modulated by $\bB$ to yield an updated probability tensor $\bPt_l$. More specifically, let $\{x_0, x_1, x_2\}$ and $\beta$ be the elements in $(\bP_l + \bDel\!\bP)$ and $\mathbf{B}$, respectively,  corresponding to a trit in $\bT_l$. Then, the corresponding updated probabilities $\{\tilde{p}_0, \tilde{p}_1, \tilde{p}_2\}$ in $\bPt_l$ are determined using the softmax function,
\begin{equation}
\tilde{p}_i=\frac{e^{\beta x_i}}{\sum_{j=0}^{2}e^{\beta x_j}}, \quad \quad i=0, 1, 2.
\label{eq:beta_limit}
\end{equation}
Intuitively, a high $\beta$ sharpens the probability mass function around the largest input, whereas a low $\beta$ flattens it. It is proven in Appendix \Appenproof \ that the entropy $H(\{\tilde{p}_0, \tilde{p}_1, \tilde{p}_2\})$ is a monotonic decreasing function of $\beta$. Thus, to reduce the entropy, we should set a large $\beta$. However, the number of required bits is not the ordinary entropy but the cross-entropy
\begin{equation}
\ell_{\textrm{CRR}} = - \sum_{i=0}^{2} q_i \log_2 \tilde{p}_i,
\label{eq:ce_loss}
\end{equation}
where $\{q_0, q_1, q_2\}$ is the ground-truth one-hot vector for the trit. If the trit corresponds to a highly complicated image region, its probabilities are hard to predict. In such a case, it is beneficial to flatten $\{\tilde{p}_0, \tilde{p}_1, \tilde{p}_2\}$ with a small $\beta$ and thus to reduce $\ell_{\textrm{CRR}}$ in \eqref{eq:ce_loss} on average.

We train CRR to minimize the sum of the cross-entropies in \eqref{eq:ce_loss} for all trits in $\bT_l$. In other words, CRR is learned to modify the input probabilities in $\bP_l$ with the additive term $\bDel\!\bP$ and then flatten or sharpen the resulting probabilities with the modulating term $\bB$, so the output probabilities in $\bPt_l$ minimize the length of the bitstream.

Figure \ref{fig:feature_visualization} shows that there are spatial redundancies in $\bY$. Hence, neighboring trits in $\bT_l$ are also correlated. Even though $\bT_l$ for a large $l$ contains more random trits, as indicated by their high entropies in $H(\bP_l)$, CRR refines their probability estimates and reduces the entropies in $H(\bPt_l)$. The entropy reduction is observed especially in simple regions, such as sky and shadow, as shown in the last column.

It is worth pointing out that CRR can be regarded as a ternary classifier, trained with the cross-entropy loss in \eqref{eq:ce_loss}, that uses the contexts to classify each trit in $\bT_l$ into one of the three classes $0$, $1$, or $2$.

\subsection{Context-Based Distortion Reduction}
\label{ssec:decoding}

CRR in Section~\ref{ssec:encoding}, as well as existing context models \cite{y2018_NIPS_minnen, y2018_CVPR_mentzer,y2021_CVPR_he}, aims to reduce the required bits for latent elements by predicting their probabilities more accurately. All these context models are used \textit{before} entropy encoding. In contrast, we propose another context model, CDR, that is used \textit{after} entropy decoding. Unlike non-progressive codecs, the proposed algorithm can use a partial latent tensor $\bYh_l$, for any $0< l \leq L$, to reconstruct the image $\bXh_l$. Thus, after decoding $\bYh_l$, which is a truncated approximation of $\bY$, CDR tries to reduce the error $\|\bY-\bYh_l\|_F$ using contexts, thereby reducing the image distortion $\|\bX - \bXh_l\|_F$ as well. Here, $\| \cdot \|_F$ denotes the Frobenius norm.

Figure \ref{fig:context_based_modules}(b) shows the architecture of CDR. Using $\bM$ and $\bSig$ as the contexts, CDR refines the partial latent tensor $\bYh_l$ into $\bYt_l$. It regresses the residual $\bDel\!\bY$ and yields the sum
\begin{equation}
\bYt_l=\bYh_l+\bDel\!\bY
\end{equation}
as the refined tensor. Note that, different from CRR, CDR does not use $\bE_l$ and $\bP_l$ as contexts, for they contain probabilistic information about $\bT_l$. Since $\bT_l$ is already decoded and used to reconstruct $\bYh_l$, $\bE_l$ and $\bP_l$ hardly provide additional information not included in $\bYh_l$. Also, note that CDR is a regressor for reducing the distortion, whereas CRR is a classifier for reducing the bitrate.

The CDR module is trained to minimize the loss
\begin{equation}
\ell_{\textrm{CDR}} = \| \bY-\bYt_l \|_F.
\label{eq:cdr_loss}
\end{equation}
For example, Figure \ref{fig:latent_refinement}(a) shows the reconstructed images from partial latent tensors $\bYh_{L-2}$, with noticeable compression artifacts. In contrast, Figure \ref{fig:latent_refinement}(b) is the reconstruction from the refined tensors $\bYt_{L-2}$ by CDR, in which the artifacts are alleviated.

\begin{figure}[t]
    \begin{center}
    \includegraphics[width=\linewidth]{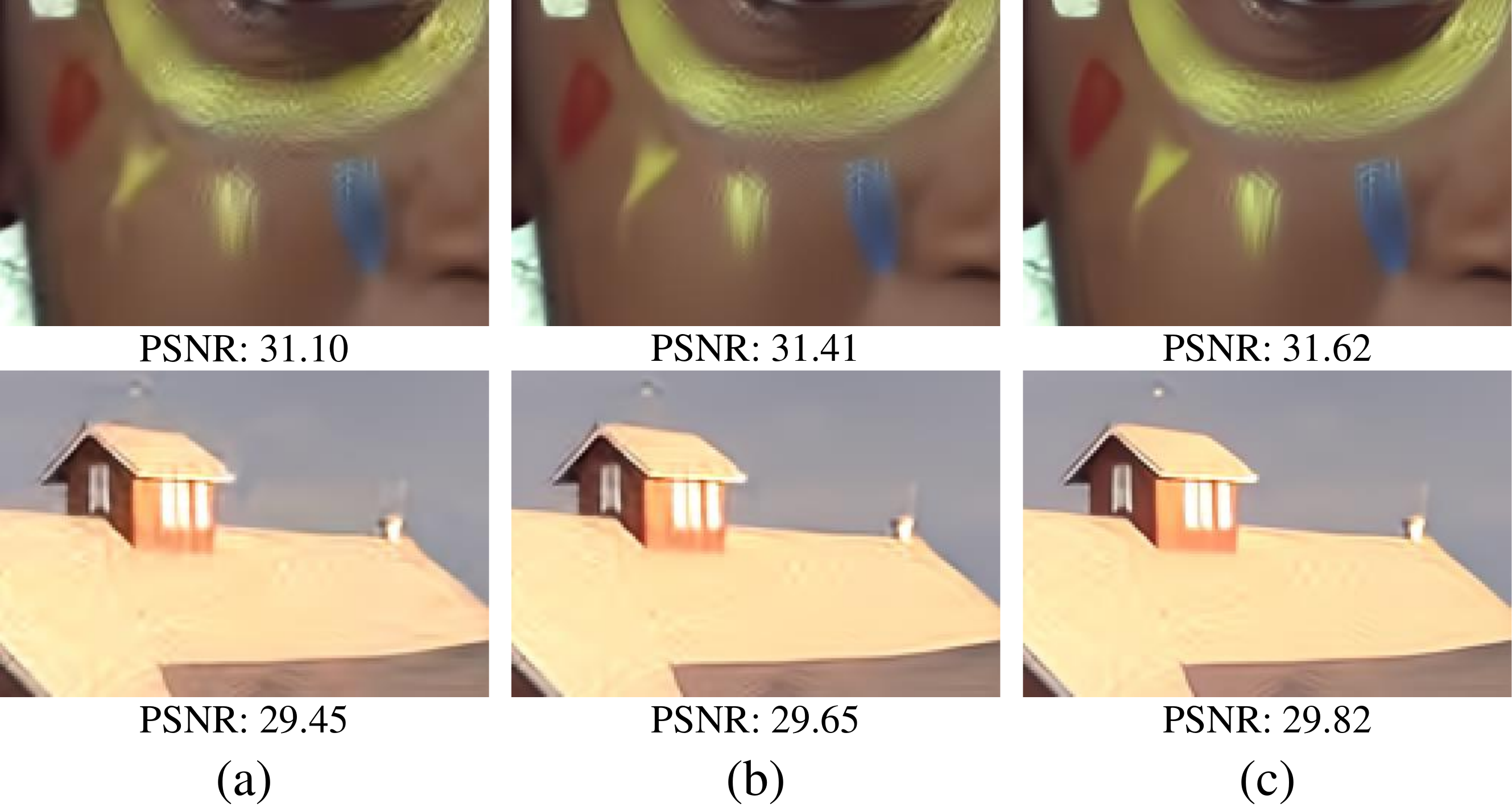}
    \end{center}
    \vspace*{-0.5cm}
    \caption
    {
        Reconstructed images (a) from partial latent tensors $\bYh_{L-2}$, (b) from refined latent tensors $\bYt_{L-2}$ and
        (c) from $\bYt_{L-2}$  with the retrained decoder.
    }
    \label{fig:latent_refinement}
\end{figure}

\subsection{Decoder Retraining}
\label{ssec:decoder_retraining}

In trit-plane coding, both the encoder and the decoder are trained for a fixed point in the RD curve (usually a high-rate, low-distortion point), and a resultant latent tensor $\bY$ is sliced into trit-planes for progressive compression \cite{y2022_CVPR_lee}. We also adopt this strategy to first train the encoder $g_a$, the hyper-encoder $h_a$, the decoder $g_s$, and the hyper-decoder $h_s$ in Figure~\ref{fig:model_architecture}. Then, we obtain $\bY$ and truncate it to various versions $\bYh_l$, $0 < l \leq L$. Using these partial tensors $\bYh_l$, we train the CRR and CDR modules, respectively, to reduce the required bitrates and the distortions by minimizing the losses in \eqref{eq:ce_loss} and \eqref{eq:cdr_loss}.

In Figure~\ref{fig:model_architecture}, trit-plane slicing and reconstruction are not differentiable, so CRR and CDR, which process trit-planes $\bT_l$ and partial tensors $\bYh_l$, cannot be trained jointly with $g_a$, $h_a$, $g_s$, and $h_s$ in an end-to-end manner. Hence, we adopt the sequential training scheme.

CDR refines $\bYh_l$ into $\bYt_l$, which is used as the new input to the decoder $g_s$. Thus, we retrain $g_s$ to further improve the quality of the reconstructed image $\bXh_l$. Specifically, we generate $\bYt_l$ for various $l$ and retrain the decoder $g_s$ to minimize
 \begin{equation}
\ell_{\textrm{DEC}} = \sum_l w_l \times \| g_s(\bYt_l) - \mathbf{X} \|_F,
\label{eq:loss_dec}
\end{equation}
where $w_l$ is a weighting parameter for each significance level $l$. The retraining improves the reconstruction quality, as illustrated in Figure~\ref{fig:latent_refinement}(c).

\section{Experiments}
\label{sec:experiments}

\subsection{Implementation and Evaluation}
\label{ssec:implementation}
We implement the proposed CTC algorithm based on the Cheng \etal's network \cite{y2020_CVPR_cheng}, composed of residual blocks and attention modules. However, we eliminate the autoregressive model and instead adopt CRR and CDR to exploit contexts. Also, we employ the unimodal Gaussian prior, rather than the Gaussian mixture model in \cite{y2020_CVPR_cheng}, to simplify the latent reconstruction in~\eqref{eq:reconstruction_level} and the computation of $\bP_l$. We use the ANS coder \cite{y2013_arXiv_duda_ANS} for the entropy coding.

\begin{figure*}[h]
    \begin{center}
    \includegraphics[width=\linewidth, height=5.5cm]{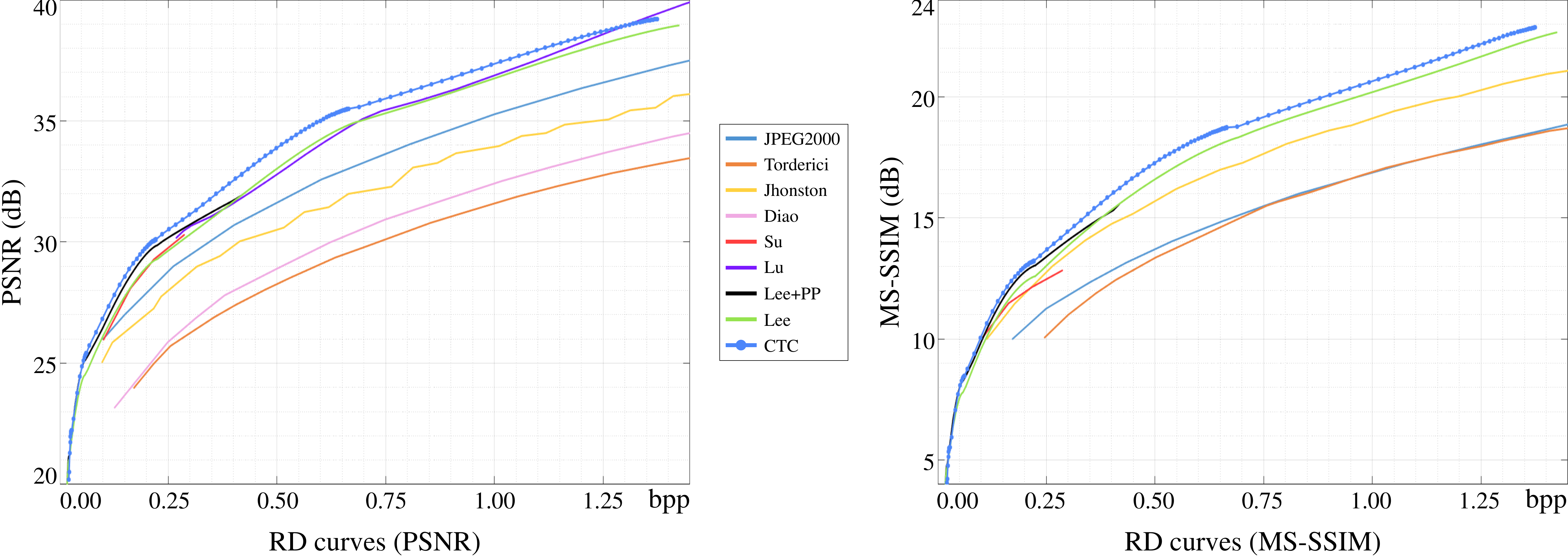}
    \end{center}
    \vspace*{-0.5cm}
    \caption
    {
        RD curve comparison of the proposed CTC algorithm with existing \textit{progressive} codecs on the Kodak lossless dataset: Torderici \etal \cite{y2017_CVPR_toderici}, Jhonston \etal \cite{y2018_CVPR_johnston}, Diao \etal \cite{y2020_DOC_diao}, Su \etal \cite{y2020_ICIP_su}, Lu \etal \cite{y2021_ICIP_lu} and Lee \etal \cite{y2022_CVPR_lee}. `+PP' means that the postprocessing networks are used to improve Lee \etal.
        The performance of JPEG2000 is measured in the default non-progressive mode to be used as the same benchmark in both this figure and Figure~\ref{fig:rdcurves_nonprogressive_kodak}.
    }
    \vspace*{-0.1cm}
    \label{fig:rdcurves_progressive_kodak}
\end{figure*}

\begin{figure}[t]
    \begin{center}
    \includegraphics[width=\linewidth,height=5.5cm]{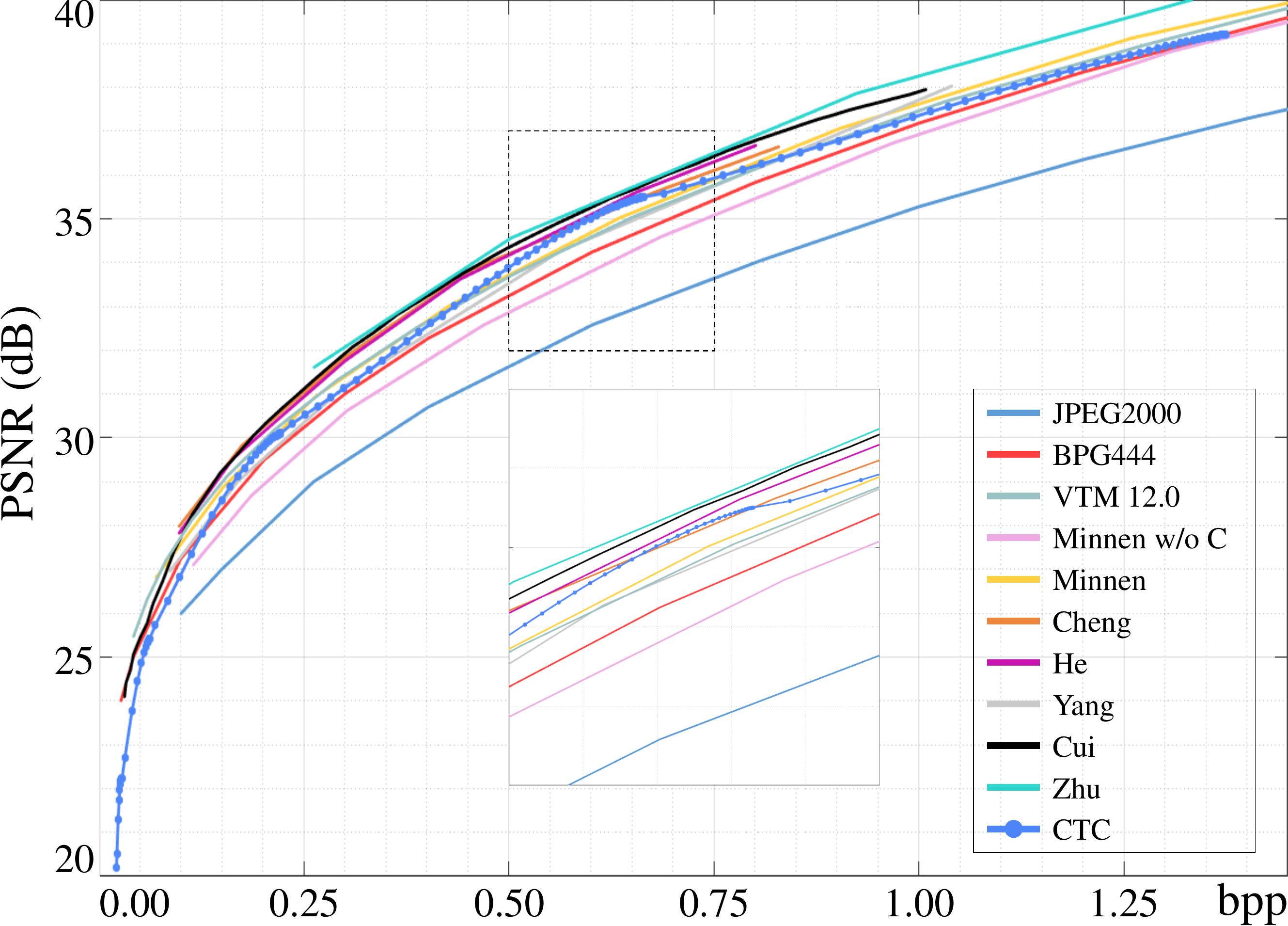}
    \end{center}
    \vspace*{-0.5cm}
    \caption
    {
        RD curve comparison of CTC with existing \textit{non-progressive} codecs on Kodak: JPEG2000 \cite{codec_jpeg2000}, BPG444 \cite{bpg}, VTM 12.0 \cite{y2021_TCSVT_Bross}, Minnen \etal \cite{y2018_NIPS_minnen}, Cheng \etal \cite{y2020_CVPR_cheng}, He~\etal \cite{y2021_CVPR_he}, Yang \etal \cite{y2021_CVPR_yang}, Cui \etal \cite{y2021_CVPR_cui}, and Zhu \etal \cite{y2022_ICLR_zhu}.
    }
    \label{fig:rdcurves_nonprogressive_kodak}
\end{figure}

To train CTC, we sample frames from the Vimeo-90K dataset \cite{y2019_IJCV_xue} and crop $256 \times 256$ patches from each frame as input. We use the Adam optimizer \cite{y2015_ICLR_kingma} with a batch size of 8 and set a learning rate of $10^{-4}$ with cosine annealing cycles \cite{y2017_ICLR_huang}. First, we train $g_a$, $h_a$, $g_s$, and $h_s$ by approximating the quantizer $q(\cdot)$ in Figure~\ref{fig:model_architecture} with the additive uniform noise model \cite{y2017_ICLR_balle}. Second, we train three sets of CRR and CDR, respectively, for different intervals of $l$. Third, we retrain the decoder $g_s$ to minimize the loss in \eqref{eq:loss_dec}. More implementation and training details are in Appendix~\Appentrainingdetail.

For evaluation, we use the Kodak lossless dataset \cite{kodim}, the CLIC professional validation dataset \cite{clic2021}, and the JPEG-AI testset\cite{jpegai}. Kodak consists of 24 images of resolution $512 \times 768$ or $768 \times 512$, while CLIC and JPEG-AI contain 41 and 16 images of up to 2K resolution. We report bitrates in bits per pixel (bpp) and measure image qualities in PSNR and MS-SSIM~\cite{y2003_ACS_wang_MS_SSIM}. For MS-SSIM, we present decibel scores by $\textrm{MS-SSIM (dB)} = - 10 \cdot \log_{10} (1 - \textrm{MS-SSIM})$. Also, we compare the compression performances of two algorithms using the BD-rate metric \cite{bdrate_excel}.

\subsection{Performance Comparison}
\label{ssec:performance_comparison}
\noindent
\textbf{RD curves:} We compare the proposed CTC algorithm with traditional BPG444 \cite{bpg}, VTM 12.0 \cite{y2021_TCSVT_Bross} and learning-based codecs in \cite{y2017_CVPR_toderici,y2018_CVPR_johnston,y2020_DOC_diao,y2020_ICIP_su,y2018_NIPS_minnen,y2020_CVPR_cheng,y2021_CVPR_he,y2021_CVPR_cui,y2021_CVPR_yang,y2021_ICIP_lu,y2022_ICLR_zhu,y2022_CVPR_lee}.

\begin{table}
    \caption
        {
            BD-rate performances (\%) with respect to Lee \etal\cite{y2022_CVPR_lee}.
        }
    \centering
    \footnotesize
    \vspace*{-0.2cm}
    \begin{tabular}{@{} l | l | r r r@{} }
        %\noalign{\smallskip}\noalign{\smallskip}\hline\hline
        \toprule
        \multicolumn{2}{c}{} & Kodak & CLIC & JPEG-AI \\
        \midrule

        & JPEG2000~\cite{codec_jpeg2000} & $31.19$ & $48.54$ & $37.46$ \\
        Traditional & BPG444~\cite{bpg} & $-12.16$ & $-1.25$ & $-7.95$ \\
        & VTM 12.0~\cite{y2021_TCSVT_Bross} & $-13.44$ & $-8.75$ & $-14.58$ \\

        \midrule

        & Minnen w/o C~\cite{y2018_NIPS_minnen} & $-8.61$ & $-0.90$ & $-4.12$ \\
        Fixed-rate & Minnen \etal~\cite{y2018_NIPS_minnen} & $-16.65$ & $-11.11$ & $-13.92$ \\
        & Cheng~\etal~\cite{y2020_CVPR_cheng} & $-23.99$ & $-15.99$ & - \quad\quad \\

        \midrule

        & Lu~\etal~\cite{y2021_ICIP_lu} & $-0.61$ & - \quad\quad & $2.91$ \\
        FGS & Lee~\etal+PP \cite{y2022_CVPR_lee} & $-6.84$ & $-6.87$ & $-7.19$ \\
        & CTC & $-14.84$ & $-14.75$ & $-17.00$ \\
        \bottomrule
    \end{tabular}
\label{table:bdrates}
\end{table}

\begin{table}
    \caption
        {
            Complexity comparison of CTC with Minnen \etal \cite{y2018_NIPS_minnen} and Lee \etal \cite{y2022_CVPR_lee}. The average encoding and decoding times for a single image in the Kodak lossless dataset are reported.
        }
    \centering
    \footnotesize
    \vspace*{-0.2cm}
    \begin{tabular}{@{} l | c c c@{} }
        \toprule
        & \# Parameters & Encoding (s) & Decoding (s) \\
        \midrule
        Minnen~\etal~\cite{y2018_NIPS_minnen} & 30.6M & 4.01 & 11.02 \\
        Lee~\etal(+PP) \cite{y2022_CVPR_lee}  & 27.2M (+50M) & 1.73 & 1.40 (+0.10)\\
        CTC  & 39.9M & 1.78 & 1.55 \\
        \bottomrule
    \end{tabular}
\label{table:times}
\end{table}

\begin{figure*}[h]
    \begin{center}
    \includegraphics[width=\linewidth]{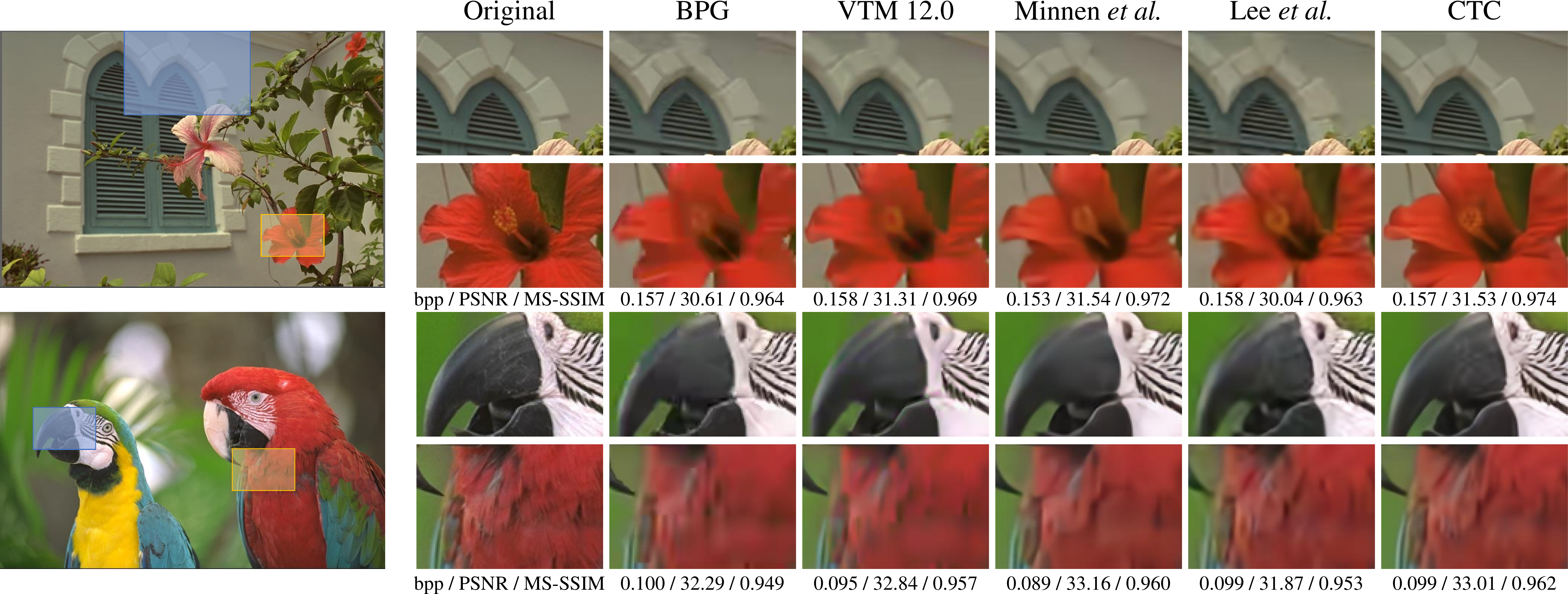}
    \end{center}
    \vspace*{-0.5cm}
    \caption
    {
        Comparison of reconstructed images by different codecs: BPG~\cite{bpg}, VTM 12.0~\cite{y2021_TCSVT_Bross}, Minnen~\etal~\cite{y2018_NIPS_minnen}, Lee~\etal~\cite{y2022_CVPR_lee} and CTC.
    }
    \vspace*{-0.1cm}
    \label{fig:qualitative_results_withcodecs}
\end{figure*}

\begin{figure*}[h]
    \begin{center}
    \includegraphics[width=\linewidth]{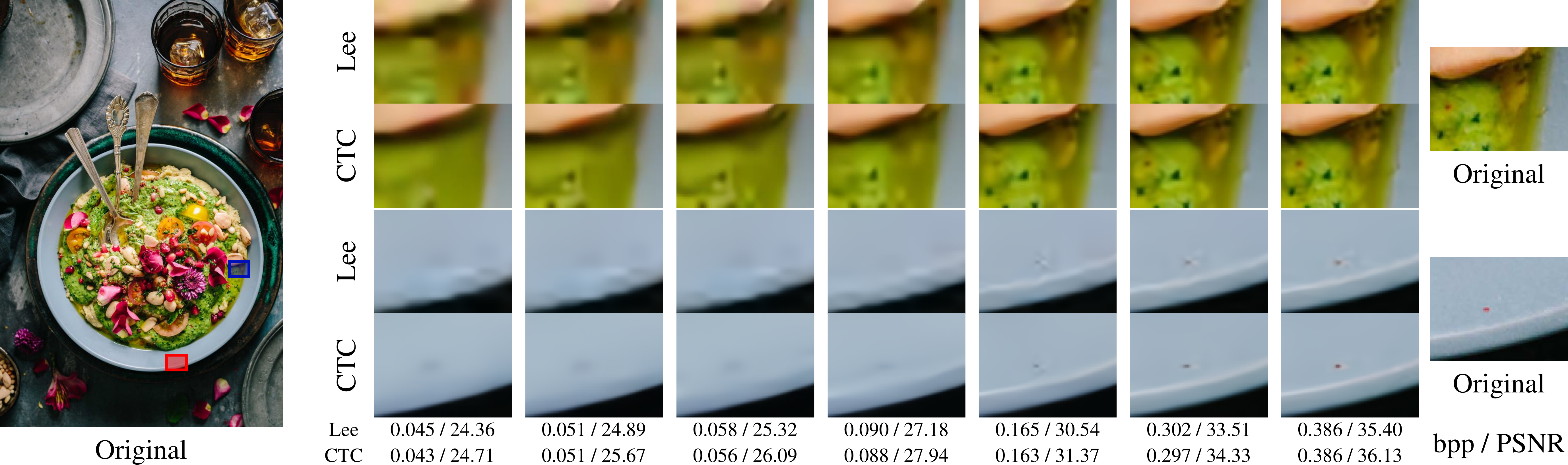}
    \end{center}
    \vspace*{-0.5cm}
    \caption
    {
        Qualitative comparison of progressive reconstruction results by Lee~\etal~\cite{y2022_CVPR_lee} and CTC. The bitrates (bpp) and PSNRs (dB) for the entire image are also listed in the corresponding columns.
    }
    \vspace*{-0.1cm}
    \label{fig:qualitative_results_progressive}
\end{figure*}

Figure~\ref{fig:rdcurves_progressive_kodak} compares the RD curves of CTC with those of progressive codecs on the Kodak lossless dataset. CTC outperforms all conventional codecs with meaningful gaps in both PSNR and MS-SSIM at a wide range of bitrates. For example, at 0.5bpp, CTC yields at least 0.8dB better PSNR than the competing codecs Lee \etal  \cite{y2022_CVPR_lee} and Lu \etal \cite{y2021_ICIP_lu} do. Notice that CTC and these two codecs support FGS. Whereas these codecs do not use any context models, CTC exploits CRR and CDR and improves the RD curves significantly. On the other hand, Su \etal \cite{y2020_ICIP_su} supports a narrow range of bitrates only, while the other learning-based codecs in \cite{y2017_CVPR_toderici,y2018_CVPR_johnston,y2020_DOC_diao} provide even worse PSNR curves than JPEG2000 \cite{codec_jpeg2000}.

Next, Figure~\ref{fig:rdcurves_nonprogressive_kodak} compares CTC with non-progressive codecs: traditional codecs \cite{codec_jpeg2000,bpg,y2021_TCSVT_Bross}, learning-based fixed-rate codecs \cite{y2018_NIPS_minnen,y2020_CVPR_cheng,y2021_CVPR_he,y2022_ICLR_zhu} and variable-rate codecs \cite{y2021_CVPR_cui,y2021_CVPR_yang}. `Minnen w/o C' means the Minnen \etal's network without the context model \cite{y2018_NIPS_minnen}. Although CTC supports the additional functionality of FGS, it yields a comparable curve to these non-progressive codecs. Especially, around 0.6bpp, CTC provides competitive PSNRs to the existing codecs, including Cui~\etal~\cite{y2021_CVPR_cui} and VTM 12.0~\cite{y2021_TCSVT_Bross}, which are the state-of-the-art variable-rate codecs. Also, CTC outperforms `Minnen w/o C' \cite{y2018_NIPS_minnen} and BPG444 \cite{bpg} at almost every bitrate. More RD curves on other datasets are available in Appendix \Appenresults.

\begin{figure*}[h]
    \begin{center}
    \includegraphics[width=\linewidth,height=5.6cm]{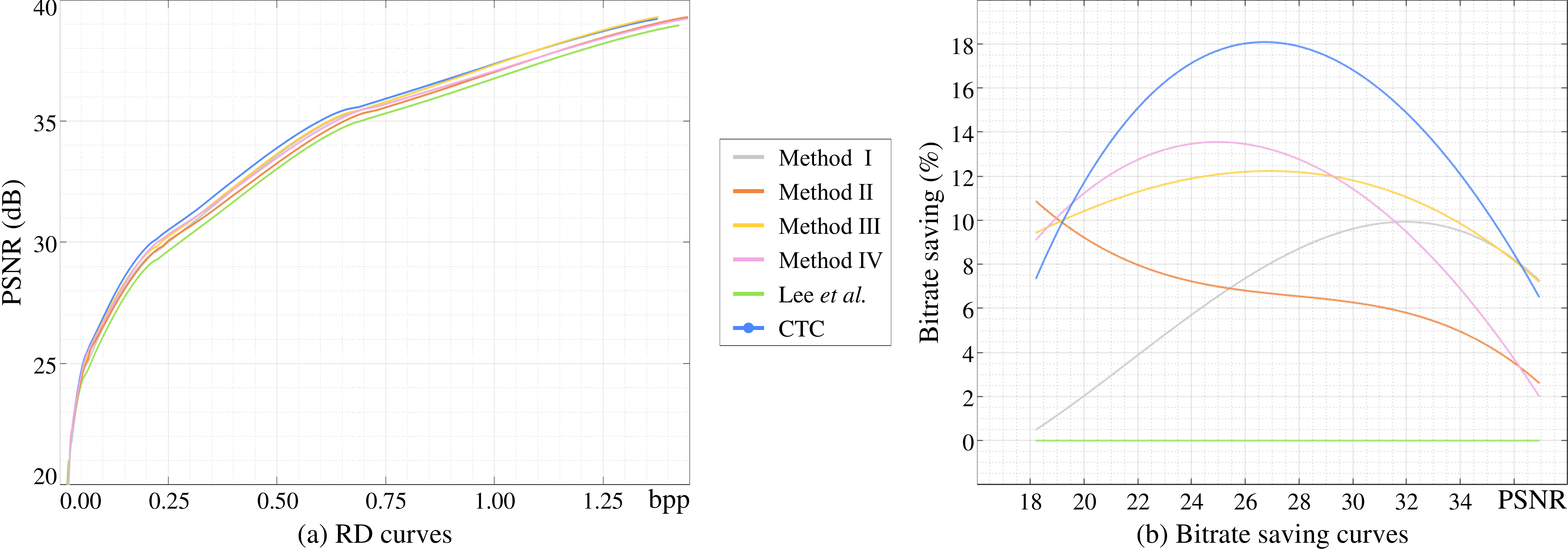}
    \end{center}
    \vspace*{-0.5cm}
    \caption
    {
        (a) RD curves of the four ablated methods in Table~\ref{table:ablation} and the baseline codec, Lee~\etal~\cite{y2022_CVPR_lee}, and (b) the correspoding bitrate saving curves with respect to the baseline.
    }
    \label{fig:ablation}
\end{figure*}

\vspace{0.1cm}
\noindent
\textbf{BD-rates:}
Table~\ref{table:bdrates} lists the BD-rates relative to Lee \etal\cite{y2022_CVPR_lee} on the three test datasets. Among the FGS codecs, the proposed CTC provides by far the best results on all datasets. For instance, on JPEG-AI, CTC achieves 17.00\% bitrate saving, while Lu \etal \cite{y2021_ICIP_lu} rather increases the required bitrates. Also, on CLIC, CTC is comparable to Cheng \etal \cite{y2020_CVPR_cheng} and better than VTM 12.0 \cite{y2021_TCSVT_Bross}.

\vspace{0.1cm}
\noindent
\textbf{Complexities: }
Table~\ref{table:times} compares the complexities of CTC with those of Minnen \etal \cite{y2018_NIPS_minnen} and Lee \etal \cite{y2022_CVPR_lee}. Minnen \etal is a fixed-rate codec using the autoregressive context model. The Lee \etal's codec supports FGS based on trit-plane coding, but it uses no context model. For Lee \etal and CTC, the times are measured for encoding and decoding an entire bitstream.

CTC is much faster than Minnen~\etal, since both CRR and CDR exploit contexts efficiently in parallel using common convolution layers, whereas Minnen \etal perform context-based prediction serially. Compared with Lee \etal, CRR and CDR demand about 12.7M more parameters but increase time complexities only marginally. In other words, CRR and CDR are not only effective for improving the RD performance but also efficient in terms of time complexity. Moreover, in Lee \etal, the postprocessing (PP) networks are optionally used to improve the reconstruction quality as shown in Figure~\ref{fig:rdcurves_nonprogressive_kodak}, but they increase the number of parameters by 50M. Without using such PP, CTC outperforms Lee \etal significantly.

\vspace{0.1cm}
\noindent
\textbf{Qualitative results}
Figure~\ref{fig:qualitative_results_withcodecs} compares reconstructed images obtained by existing codecs \cite{bpg,y2021_TCSVT_Bross,y2018_NIPS_minnen,y2022_CVPR_lee} and CTC. Near sharp edges or in textured regions, such as the window and wall patterns, flowers, and feathers, the traditional codecs \cite{bpg, y2021_TCSVT_Bross} yield blur artifacts. The reconstruction quality of the proposed CTC is better than that of Lee \etal \cite{y2022_CVPR_lee} and is comparable to that of the Minnen \etal's non-progressive codec \cite{y2018_NIPS_minnen}.

Figure~\ref{fig:qualitative_results_progressive} compares progressive reconstruction results, obtained by Lee~\etal~\cite{y2022_CVPR_lee} and CTC. At each column, both trit-plane coding algorithms reconstruct the images $\bXh_l$ up to the same significance level $l$. The proposed CTC yields higher RD performances by employing context models and decoder retraining. Consequently, CTC provides a better image quality than Lee \etal does.% Especially, in areas with the availability of utilizing contexts, there is a significant qualitative difference even in high PSNR regions.

\subsection{Ablation Study}
\label{ssec:ablation}

We conduct an ablation study to analyze the three contributions --- CRR, CDR, and decoder retraining --- of the proposed CTC algorithm as compared with the baseline trit-plane codec, Lee \etal \cite{y2022_CVPR_lee}. Table~\ref{table:ablation} lists the BD-rates of four ablated methods relative to the baseline on the Kodak dataset. CRR and CDR in methods \RNum{1} and \RNum{2} improve the RD performances, respectively, by reducing bitrates and improving image qualities. Both CRR and CDR achieve about $7\%$ of bitrate saving. When they are used together, the bitrate saving in method \RNum{3} is as big as $10.93\%$. Also, the decoder retraining with CDR provides a similar reduction of $10.81\%$, indicating that the retraining for partial latent tensors $\bYh_l$ is also essential in trit-plane coding. By combining the three components, the proposed CTC algorithm achieves a significant bitrate saving of $14.84\%$.

Figure~\ref{fig:ablation}(a) compares the RD curves of the ablated methods in Table~\ref{table:ablation}, and Figure~\ref{fig:ablation}(b) plots the bitrate saving percentages in terms of PSNR with respect to the baseline. We see that method \RNum{1} is more effective at a high PSNR range, since trit probabilities can be more accurately predicted using contexts when latent elements are finely reconstructed. On the other hand, method \RNum{2} performs better in a low PSNR range because quantization noise of coarsely reconstructed latent elements can be more easily reduced. The method \RNum{3} exhibits a relatively even bitrate saving in the entire PSNR range. Method \RNum{4} yields a bitrate saving curve skewed to low PSNRs. Finally, CTC reduces the bitrate requirement significantly, by more than $10\%$, when PSNR is between 20dB and 35dB. Therefore, the whole bitrate saving is $14.84\%$ as listed in Table \ref{table:ablation}.

\begin{table}
    \caption
        {
            Ablation study of CTC: for each ablated method, the BD-rate relative to the baseline, Lee \etal \cite{y2022_CVPR_lee}, is reported.
        }
    \vspace*{-0.2cm}
    \centering
    \footnotesize
    \begin{tabular}{@{} l | c c c | r@{} }
        \toprule
        & CRR & CDR & $g_s$ retraining & BD-rate \\
        \midrule
        Method \RNum{1} & \ding{51} & - & - & $-6.90\%$ \\
        Method \RNum{2} & - & \ding{51} & - & $-6.68\%$ \\
        Method \RNum{3} & \ding{51} & \ding{51} & - & $-10.93\%$ \\
        Method \RNum{4} & - & \ding{51} & \ding{51} & $-10.81\%$ \\
        \midrule
        CTC & \ding{51} & \ding{51} & \ding{51} & $\mathbf{-14.84\%}$ \\
        \bottomrule
    \end{tabular}
\label{table:ablation}
\end{table}

% \subsection{Limitations}
% \label{ssec:limitation}
% We introduced the context models CRR and CDR for deep progressive image compression for the first time, but there are limitations that the training of the main network, that of context models, and the decoder retraining are separated. They cannot be conducted in an end-to-end manner because trit-plane slicing and trit-plane reconstruction are not differentiable. It is a future research issue to develop differentiable approximations of the slicing and reconstruction modules and to investigate whether the end-to-end training would further improve the RD performance.

\section{Conclusions}
We proposed an effective trit-plane codec, called CTC, for progressive image compression using the two context modules: CRR and CDR. Before entropy encoding, CRR updates a probability tensor to compress trit-planes more compactly. After entropy decoding, CDR refines a  partial latent tensor to reconstruct a higher-quality image. Both CRR and CDR are based on convolutional layers, so they are efficient in terms of time complexity. Moreover, we developed a decoder retraining scheme, which, combined with CDR, achieves  better RD tradeoffs. It was shown that CTC outperforms conventional progressive codecs greatly.

\section*{Acknowledgments}
This work was supported by the National Research Foundation of Korea (NRF) grants funded by the Korea government (MSIT) (No.~NRF-2021R1A4A1031864 and No.~NRF-2022R1A2B5B03002310), and by IITP grant funded by the Korea government (MSIT) (No.~2021-0-02068, Artificial Intelligence Innovation Hub).

%%%%%%%%% REFERENCES
{\small
\bibliographystyle{ieee_fullname}
\bibliography{2023_CVPR_SMJEON}
}

\clearpage

\appendix

\section{Softmax and Entropy}
\label{app:proof}

\begin{theorem}
$H\left(\{p_0, p_1, p_2\}\right)$ is a monotonic decreasing function of $\beta$, where
\begin{equation}
p_i=\frac{e^{\beta y_i}}{\sum_{j=0}^{2}e^{\beta y_j}}, \quad \quad i=0, 1, 2,
\label{eq:softmax}
\end{equation}
and $\beta > 0$.
\end{theorem}

\begin{proof}

Let $x = e^\beta$ and $A = x^{y_0}+x^{y_1}+x^{y_2}$. Then,
\begin{align}
-H\left(\{p_0, p_1, p_2\}\right) &= -H\left(\left\{\frac{x^{y_0}}{A}, \frac{x^{y_1}}{A}, \frac{x^{y_2}}{A}\right\}\right) \\
= \frac{x^{y_0}}{A}\log{\frac{x^{y_0}}{A}} & + \frac{x^{y_1}}{A}\log{\frac{x^{y_1}}{A}} + \frac{x^{y_2}}{A}\log{\frac{x^{y_2}}{A}}.
\label{eq:entropy}
\end{align}
The derivative of $-H$ with respect to $x$ is given by
\begin{align}
-\frac{\partial H}{\partial x}
&= \sum_{i=0}^{2}\left(1+\log\frac{x^{y_i}}{A}\right)  \frac{y_i x^{y_0 - 1} A - x^{y_i}A'}{A^2} \\
&= \sum_{i=0}^{2} {\log\frac{x^{y_i}}{A} \times \frac{y_i x^{y_0 - 1} A - x^{y_i}A'}{A^2}} \label{eq:simp1} \\
&= \sum_{i=0}^{2} {\left(\log{x^{y_i}} - \log{A}\right) \frac{y_i x^{y_0 - 1} A - x^{y_i}A'}{A^2}} \\
&= \sum_{i=0}^{2} {\log{x^{y_i}} \times \frac{y_i x^{y_0 - 1} A - x^{y_i}A'}{A^2}}
\label{eq:simp2}\\
&= \sum_{i=0}^{2} {\log{x} \times \frac{{y_i}\left(y_i x^{y_0 - 1} A - x^{y_i}A'\right)}{A^2}}
\label{eq:derivative_of_H}
\end{align}
where
\begin{equation}
A'=\frac{\partial A}{\partial x} = y_{0}x^{y_0 - 1}+y_{1}x^{y_1 - 1}+y_{2}x^{y_2 - 1}.
\label{eq:derivative_of_A}
\end{equation}
Note that
\begin{align}
({y_0}x^{y_0 - 1}  +{y_1}x^{y_1 - 1}&+{y_2}x^{y_2 - 1})A \nonumber \\
                  &= (x^{y_0}+x^{y_1}+x^{y_2})A'
\label{eq:AAprime}
\end{align}
and the equalities in \eqref{eq:simp1} and \eqref{eq:simp2} come from \eqref{eq:AAprime}.

Then, we have
\begin{align}
-\frac{\partial H}{\partial x}\frac{A^2}{\log x}
&= \sum_{i=0}^{2} {{y_i}^2 x^{{y_i}-1}A - {y_i}x^{y_i}A'} \\
&= x^{y_0 + y_1 - 1}\left(y_0 - y_1\right)^2 \\
&+ x^{y_1 + y_2 - 1}\left(y_1 - y_2\right)^2 \\
&+ x^{y_2 + y_0 - 1}\left(y_2 - y_0\right)^2 \\
\label{eq:geqzero}
&\geq 0.
\end{align}
\noindent
Thus, if $x>1$, then $\frac{\partial H}{\partial x}\leq 0$ and $H$ is a strictly monotonic decreasing function of $x$ unless $y_0=y_1=y_2$. Moreover, $x=e^\beta$ is a strictly monotonic increasing function of $\beta$, and $x>1$ if $\beta >0$. Therefore, $H$ is a strictly monotonic decreasing function of $\beta$, provided that $\beta > 0$.
\end{proof}

\section{Implementation and Training Details}
\label{app:imp}
In this section, we describe the implementation and training details of CTC. First, we describe the software for traditional codecs and the libraries for learning-based algorithms. Second, we present the implementation details of the proposed context models CRR and CDR. Then, we explain how to train the proposed CTC algorithm. Note that the implementation and acceleration details of DPICT\cite{y2022_CVPR_lee} are available in \cite{y2022_arXiv_jeon}.

\subsection{Software and Libraries}
\label{sapp:software}

We adopt the traditional codecs JPEG2000 \cite{codec_jpeg2000}, BPG444 \cite{bpg}, VTM 12.0 \cite{y2021_TCSVT_Bross} for comparison.

\vspace{0.2cm}
\noindent
\textbf{JPEG2000:} We use the open software in \cite{codec_jpeg2000}. We execute the following commands for encoding and decoding. We transform RGB-formatted images, such as png files, into raw files.

\vspace{0.3cm}
\noindent
\texttt{\{buildpath\}/opj\_compress -i \{inputfile\} -o \{bin\} -r \{15:150\}\\-F \{width\},\{height\},3,8,u@1x1:1x1:1x1}\vspace{0.3cm}
\\
\texttt{\{buildpath\}/opj\_decompress -i \{bin\} -o \{outputfile\}}

\vspace{0.3cm}
\noindent
\textbf{BPG444:} We use the software in \cite{bpg} and enter the following commands.
\vspace{0.3cm}

\noindent
\texttt{\{buildpath\}/bpgenc \{inputfile\} -o \{bin\} -q \{26:52\}\\-f 444 -e x265}\vspace{0.3cm}
\\
\texttt{\{buildpath\}/bpgdec -o \{outputfile\} \{bin\}}

\vspace{0.3cm}
\noindent
\textbf{VTM 12.0:} We execute the reference software package in \href{https://vcgit.hhi.fraunhofer.de/jvet/VVCSoftware_VTM/-/tree/VTM-12.0}{https://vcgit.hhi.fraunhofer.de/jvet/VVCSoftware\_VTM/-/tree/VTM-12.0} with the following commands.
\vspace{0.3cm}

\noindent
\texttt{\{buildpath\}/EncoderApp -i \{inputfile\} -c \{cfgpath\}/encoder\_intra\_vtm.cfg\\
-o /dev/null -b \{bin\} -wdt \{width\} -hgt \{height\} -fr 1 -f 1\\
-q {24, 26, 30, 31:43} --InputChromaFormat=444 --InputBitDepth=8\\
--ConformanceWindowMode=1 --InputColourSpaceConvert=RGBtoGBR\\
--SNRInternalColourSpace=1 --OutputInternalColourSpace=0
}\vspace{0.3cm}
\\
\texttt{\{buildpath\}/DecoderApp -b \{bin\} -o \{outputfile\} -d 8\\
--OutputColourSpaceConvert=GBRtoRGB}

\vspace{0.3cm}
We use \texttt{Pytorch}\cite{y2019_NIPS_paszke_pytorch} and \texttt{CompressAI}\cite{y2020_arXiv_begaint_compressai} libraries to implement the proposed CTC algorithm. Also, we employ the source codes and pretrained parameters in \texttt{CompressAI} for the Minnen \etal's algorithm \cite{y2018_NIPS_minnen}. For the other learning-based codecs, we use the results provided in the original papers.

\subsection{Implementation of CRR and CDR}
\label{sapp:context}
The main network of the proposed CTC algorithm is in Figure~\ref{fig:model_architecture}, and the detailed structures of the CRR and CDR modules are in Figure~\ref{fig:context_based_modules}. The context modules are incorporated into the main network as follows. There are three CRR models for different intervals of trit-plane levels $l$. We denote them as $\textrm{CRR}_L$, $\textrm{CRR}_{L-1}$, and $\textrm{CRR}_{\leq L-2}$, where the subscripts indicate the ranges of trit-plane levels in which the corresponding models are used. In other words,
\begin{equation}
\bPt_l =
\begin{cases}
\textrm{CRR}_L\left(\bYh_{l-1}, \bM, \bSig, \bE_l, \bP_l\right) & \text{if}\ \ l=L, \\
\textrm{CRR}_{L-1} \left(\bYh_{l-1}, \bM, \bSig, \bE_l, \bP_l\right) & \text{if}\ \ l=L - 1, \\
\textrm{CRR}_{\leq L-2} \left(\bYh_{l-1}, \bM, \bSig, \bE_l, \bP_l\right) & \text{if}\ \ l \le L - 2.
\end{cases}
\label{eq:crr_imp}
\end{equation}
Similarly, we implement three CDR models $\textrm{CDR}_{L-1}$, $\textrm{CDR}_{L-2}$, and $\textrm{CDR}_{\leq L-3}$ to obtain
\begin{equation}
\bYt_l =
\begin{cases}
\bYh_l & \text{if}\ \ L-1 < l \le L, \\
\textrm{CDR}_{L-1}\left(\bYh_{l}, \bM, \bSig\right) & \text{if}\ \ L-2 < l \le L - 1, \\
\textrm{CDR}_{L-2}\left(\bYh_{l}, \bM, \bSig\right) & \text{if}\ \ L-3 < l \le L - 2, \\
\textrm{CDR}_{\leq L-3} \left(\bYh_{l}, \bM, \bSig\right) & \text{if}\ \ l \le L - 3.
\end{cases}
\label{eq:cdr_imp}
\end{equation}

Whereas CRR estimates the probability tensor $\bPt_l$ for each trit-plane $\bT_l$ ( $l=1,\ldots, L$), CDR performs the prediction of $\bYt_l$ for any $l \leq L-1$. Therefore, $l$ is an integer in \eqref{eq:crr_imp} but a real number in \eqref{eq:cdr_imp}. The trit-plane levels $l \leq L-2$ are supported by a single CRR model, and the levels $l\leq L-3$ are by a single CDR model. These choices are made to strike a balance between the number of parameters and the RD performance. Also, CDR is not used at the top level $L-1 <  l \leq L$ because the refinement of a latent tensor $\bYh_l$ is not necessary at such a fine level.

Note that the proposed CDR is conceptually similar to LRP in \cite{y2020_ICIP_minnen}. However, there are clear differences between them. Whereas LRP predicts residual errors by taking only the mean and latent tensors as input, CDR exploits the standard deviation $\mathbf{\Sigma}$ as the additional context. In this way, CDR can refine partially reconstructed latent elements by exploiting their uncertainty levels, which are inversely proportional to the standard deviations. To demonstrate the importance of $\mathbf{\Sigma}$, we have implemented an ablated version of CDR without $\mathbf{\Sigma}$. It increases the BD-rate by
$+2.45\%$ on the Kodak lossless dataset. Moreover, while the quantization step size is $1$ in LRP, it is larger in the proposed algorithm (\ie $3^{L-l}$ when the first $l$ trit-planes are decoded). Thus, there are more opportunities for reducing quantization errors in the proposed algorithm. To achieve this goal, the proposed CDR exploits both $\mathbf{M}$ and $\mathbf{\Sigma}$.

\begin{table*}[t]
    \caption
        {
            The numbers of epochs for the context model training and the decoder retraining ($g_s$).
        }
    \centering
    \footnotesize
    \vspace*{-0.2cm}
    \begin{tabular}{c c c c c c c c@{} }
        \toprule
         $\textrm{CRR}_L$ & $\textrm{CRR}_{L-1}$ & $\textrm{CRR}_{\leq L-2}$ & $\textrm{CDR}_{L-1}$ & $\textrm{CDR}_{L-2}$ & $\textrm{CDR}_{\leq L-3}$ & $g_s$ \\
        \midrule
         300 & 300 & 10 & 30 & 30 & 30 & 100\\
        \bottomrule
    \end{tabular}
\label{table:context_models_epoch}
\end{table*}

\subsection{Training of CTC}
\label{ssec:training_details}
We train the main network for 300 epochs using the rate-distortion loss $\mathcal{L}=\mathcal{D}+\lambda \mathcal{R}$ with $\lambda=5$.

% Description of trit-plane level
In the trit-plane slicing module in Figure~\ref{fig:model_architecture}, a latent tensor is sliced into $L$ trit-planes. Note that the maximum trit-plane level $L$ depends on the latent tensor, as described in \cite{y2022_arXiv_jeon}. However, $L=7$ for most images. The selection of CRR and CDR models in \eqref{eq:crr_imp} and \eqref{eq:cdr_imp} is dependent on $L$. Therefore, for stable training of these models, as well as the decoder retraining, we fix $L=7$ and use the training images with $L=7$ only.

We use the cross-entropy loss in \eqref{eq:ce_loss} to train the three CRR models. The CRR process is performed for every trit, except when the original probabilities are $(p_0, p_1, p_2)=(0, 1, 0)$. In such a case, the trit requires no bit, and there is no reason to update its probabilities.

The CDR loss in \eqref{eq:cdr_loss} can be rewritten as
\begin{equation}
\ell_{\textrm{CDR}}(l) = \| \bY-\bYt_l \|_F
\end{equation}
where $l$ denotes a trit-plane level. The first CDR model $\textrm{CDR}_{L-1}$ in \eqref{eq:cdr_imp} supports a partially reconstructed level $l \in (L-2, L-1]$. For its training, we use the sum of losses, given by
\begin{equation}
\ell_{\textrm{CDR}}(L-1) + \ell_{\textrm{CDR}}(\alpha) + \ell_{\textrm{CDR}}(L-2)
\label{eq:cdr_loss_specific}
\end{equation}
where $\alpha \sim \mathcal{U}(0,1)$ is a uniform random variable. The losses for the other two models  $\textrm{CDR}_{L-2}$ and $\textrm{CDR}_{\leq L-3}$ are similarly defined.

Then, the decoder is retrained to minimize the loss $\ell_{\textrm{DEC}}$ in \eqref{eq:loss_dec}. Note that the original decoder is optimized for the case when all trit-planes are received (\ie the highest level $l=L$). Thus, the decoder is retrained to consider lower levels as well. However, due to the retraining, the performances at high levels can be degraded. To alleviate the degradation, we set large weighting parameters at high levels, compared to low levels. Specifically, we define the loss as
\begin{align}
\ell_{\textrm{DEC}} = 100 \times & \textstyle  \sum_{l=L-1}^{L}  \| g_s(\bYt_{l}) - \mathbf{X} \|_F \\
                    + & \textstyle \sum_{l=L-4}^{L-2} \| g_s(\bYt_{l}) - \mathbf{X} \|_F.
\end{align}
Note that we consider five levels from $L-4$ to $L$, and set bigger weights at the two highest levels $L$ and $L-1$.

The training epochs for the context models and the decoder retraining are summarized in Table \ref{table:context_models_epoch}. These training schedules are determined by observing the convergence of the validation performance.

\section{More Experiments}
\label{app:exp}

\subsection{RD curves}
\label{ssec:rd_curves}
Figures \ref{fig:rdcurves_clic} and \ref{fig:rdcurves_jpegai} compare the RD curves on the CLIC validation dataset and the JPEG-AI testset, respectively. All learning-based algorithms, including the proposed CTC, are optimized to minimize the MS-SSIM loss in Figure~\ref{fig:rdcurves_clic}(b) and Figure~\ref{fig:rdcurves_jpegai}(b).

Figure \ref{fig:rdcurves_sorting_ablation} compares the proposed CTC with the trit-plane coding without RD priorities. More specifically, in `Without RD priorities,' the trits in each trit-plane are transmitted in the 3D raster scan order, instead of the decreaing order of their RD priorities \cite{y2022_CVPR_lee}. This alternative method performs badly compared to CTC. However, we see that its performance is also improved by employing the two context modules, CRR and CDR, and the decoder retraining scheme.

\subsection{Time complexity for high-resolution images}
\label{ssec:time_complexity}
We compare the time complexities for compressing 2K images in the CLIC validation dataset and the JPEG-AI testset in Table \ref{table:time_complexity_appendix}. The proposed CTC algorithm is based on trit-plane coding, which represents each latent element with about 7 trits. Hence, CTC increases the number of entropy-coded data by a factor of about 7, as compared with non-FGS codecs such as He \etal \cite{y2021_CVPR_he}. This is the main reason (and a price for enabling FGS) that CTC is slower than \cite{y2021_CVPR_he}. However, it can be observed from Table~\ref{table:time_complexity_appendix} that the proposed context modules, CRR and CDR, increase the complexities only moderately.

\begin{table}[h]
    \centering
    \footnotesize
    \caption
        {
            Time complexity comparison of CTC with Minnen \etal \cite{y2018_NIPS_minnen} and He \etal \cite{y2021_CVPR_he}.
        }
    %\addtolength{\tabcolsep}{1.0pt}
    % \renewcommand{\arraystretch}{0.9}
    \vspace*{-0.2cm}
    \begin{tabular}{@{} l | c c@{} }
        \toprule
        & Encoding (s) & Decoding (s) \\
        \midrule
        He \etal \cite{y2021_CVPR_he} & 1.00 & 0.91 \\
        Minnen \etal \cite{y2018_NIPS_minnen} & 25.85 & 78.49 \\
        \midrule
        CTC w/o context modules & 8.10 & 7.19 \\
        CTC & 8.70 & 8.26 \\
        \bottomrule
    \end{tabular}
    \label{table:time_complexity_appendix}
\end{table}

\subsection{Reconstructed images of CTC}
\label{ssec:qualitative}
Figures \ref{fig:qualitative_kodak_1}$\sim$\ref{fig:qualitative_jpegai_2} show various images reconstructed by the proposed CTC algorithm at levels $l=L$, $L-2$, $L-3$, and $L-4$. The images with resolutions larger than $512 \times 768$ are center-cropped.

\begin{figure*}[t]
    \begin{center}
    \includegraphics[width=\linewidth]{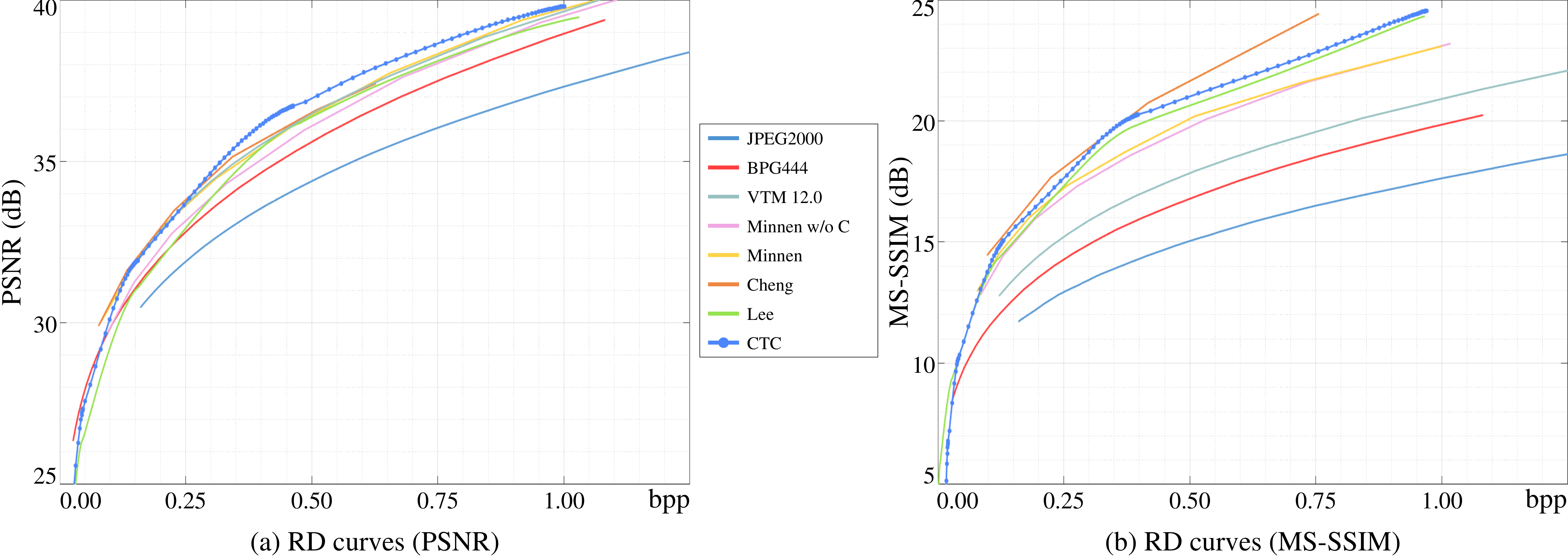}
    \end{center}
    \vspace*{-0.5cm}
    \caption
        {
            RD curve comparison on the CLIC validation dataset.
        }
    \label{fig:rdcurves_clic}
\end{figure*}

\begin{figure*}[t]
    \begin{center}
    \includegraphics[width=\linewidth]{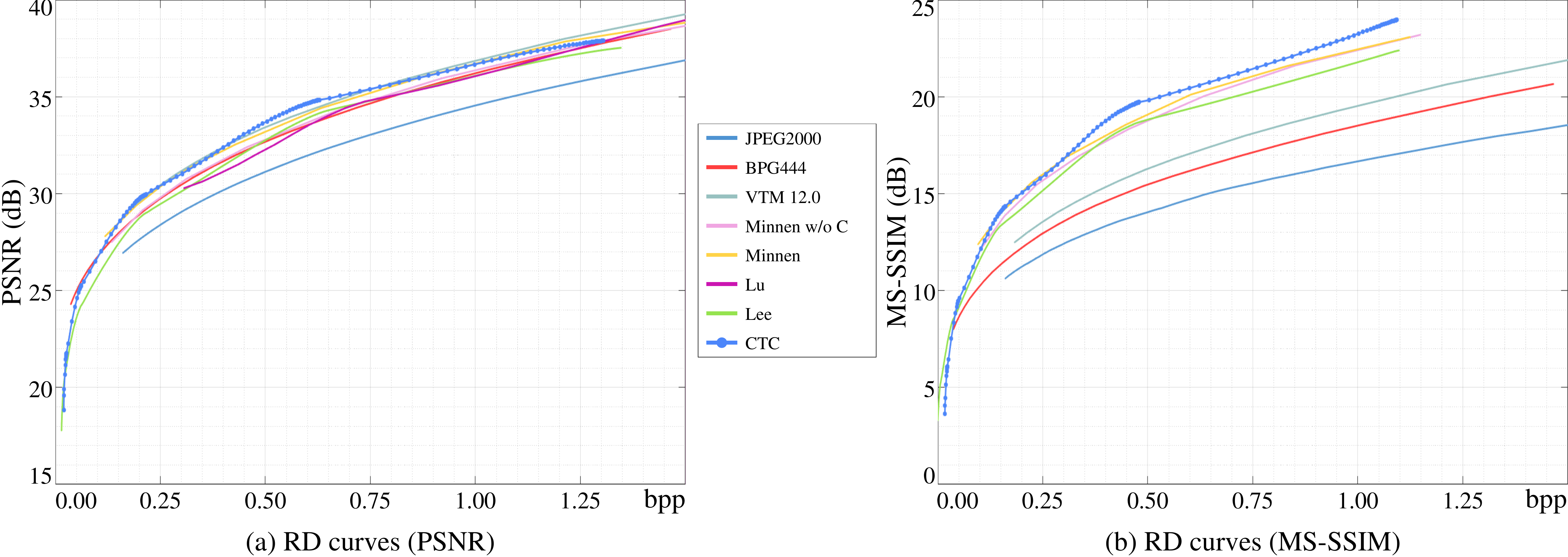}
    \end{center}
    \vspace*{-0.5cm}
    \caption
        {
            RD curve comparison on the JPEG-AI testset.
        }
    \label{fig:rdcurves_jpegai}
\end{figure*}

\begin{figure*}[t]
    \begin{center}
    \includegraphics[width=\linewidth]{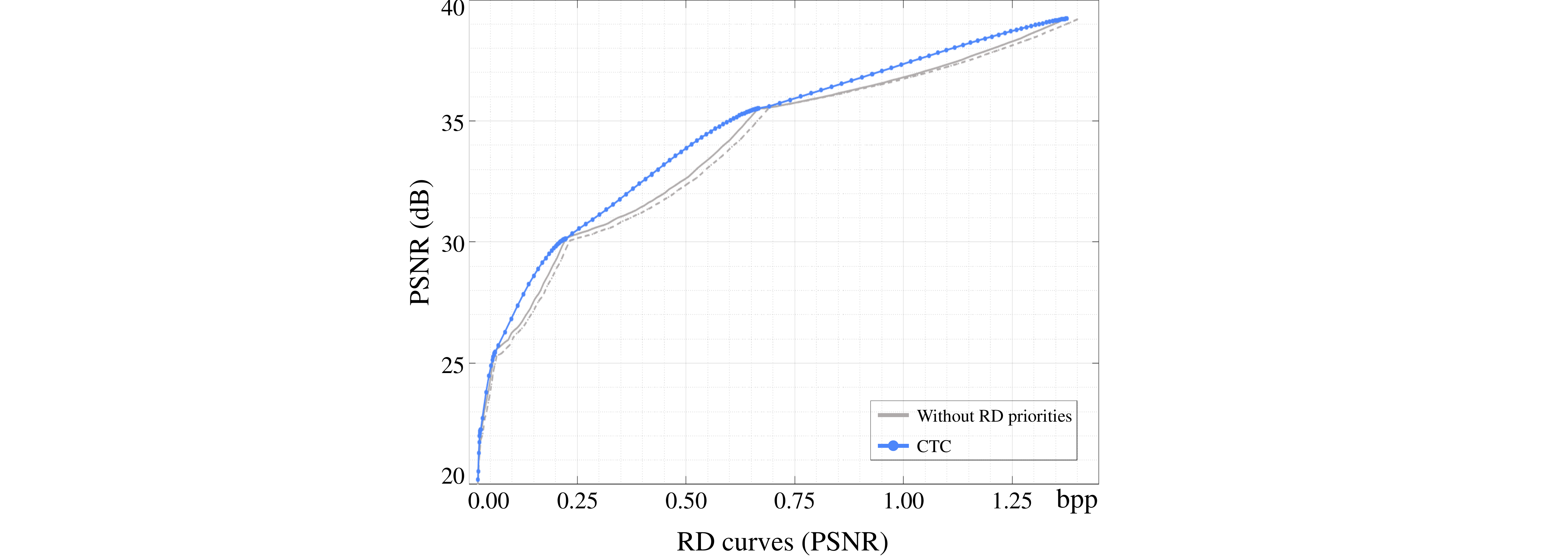}
    \end{center}
    \vspace*{-0.5cm}
    \caption
        {
            RD curve comparison of CTC and the alternative trit-plane coding method without RD priorities on the Kodak lossless dataset. The dashed curve means that the context modules, CRR and CDR, and the decoder retraining are not employed.
        }
    \label{fig:rdcurves_sorting_ablation}
\end{figure*}

\clearpage

\begin{figure*}[b]
    \begin{center}
    \includegraphics[width=\linewidth]{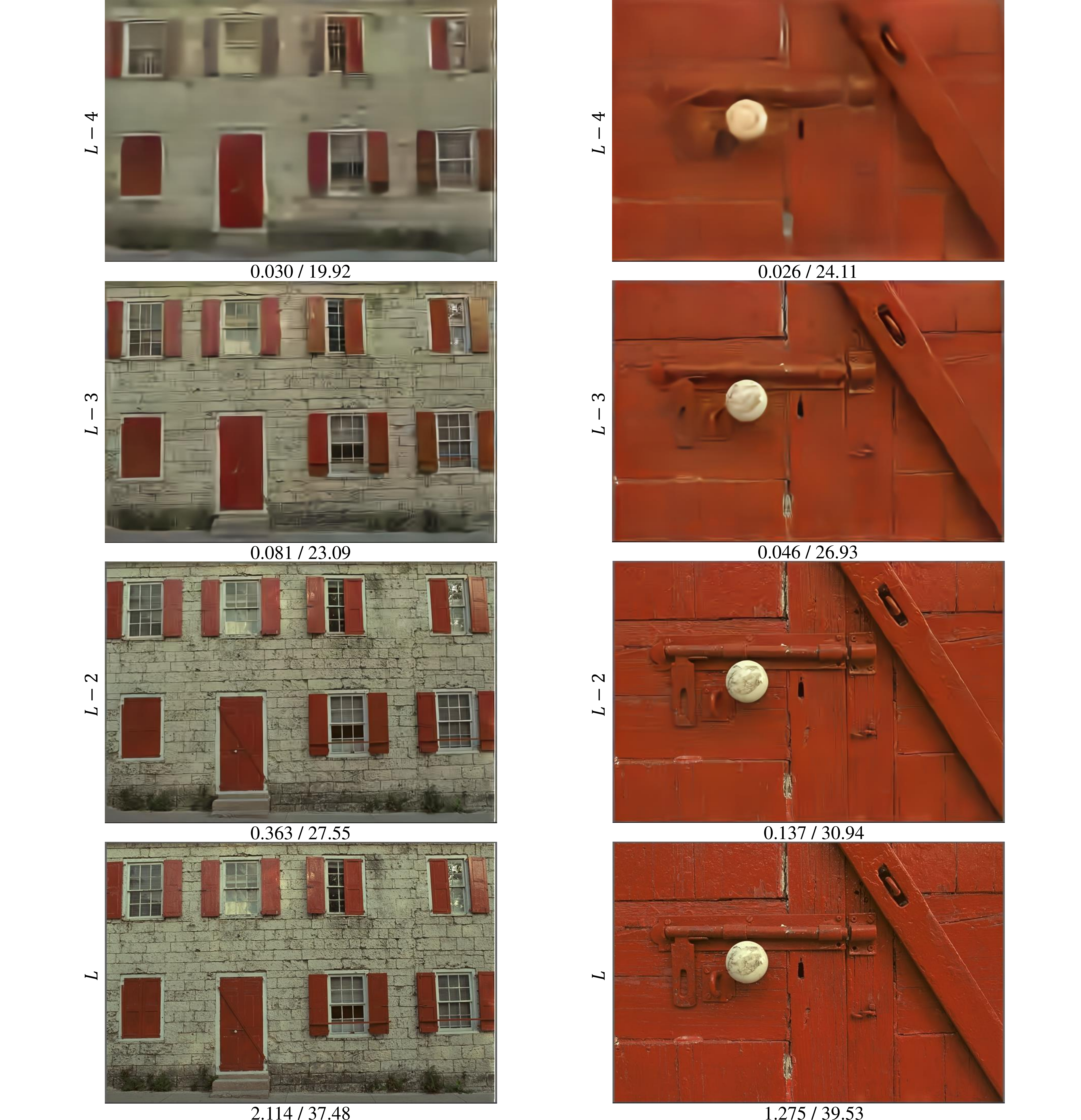}
    \end{center}
    \vspace*{-0.5cm}
    \caption
        {
            Reconstructed images of ``kodim01.png'' and ``kodim02.png.'' The bitrates (bpp) and PSNRs (dB) are reported below each image.
        }
    \label{fig:qualitative_kodak_1}
\end{figure*}

\begin{figure*}[b]
    \begin{center}
    \includegraphics[width=\linewidth]{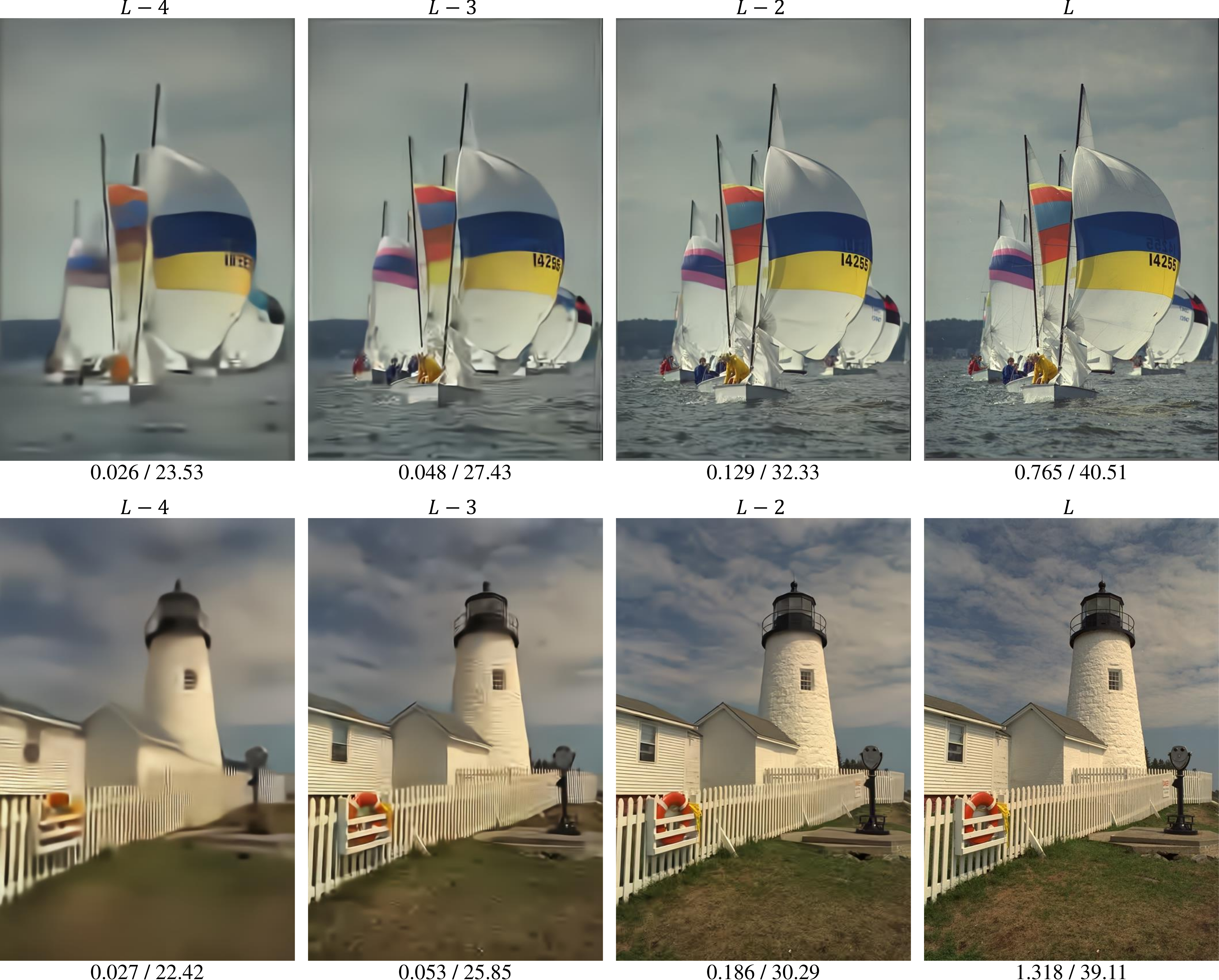}
    \end{center}
    \vspace*{-0.5cm}
    \caption
    {
        Reconstructed images of ``kodim09.png'' and ``kodim19.png.'' The bitrates (bpp) and PSNRs (dB) are reported below each image.
    }
    \label{fig:qualitative_kodak_2}
\end{figure*}

\begin{figure*}[b]
    \begin{center}
    \includegraphics[width=\linewidth]{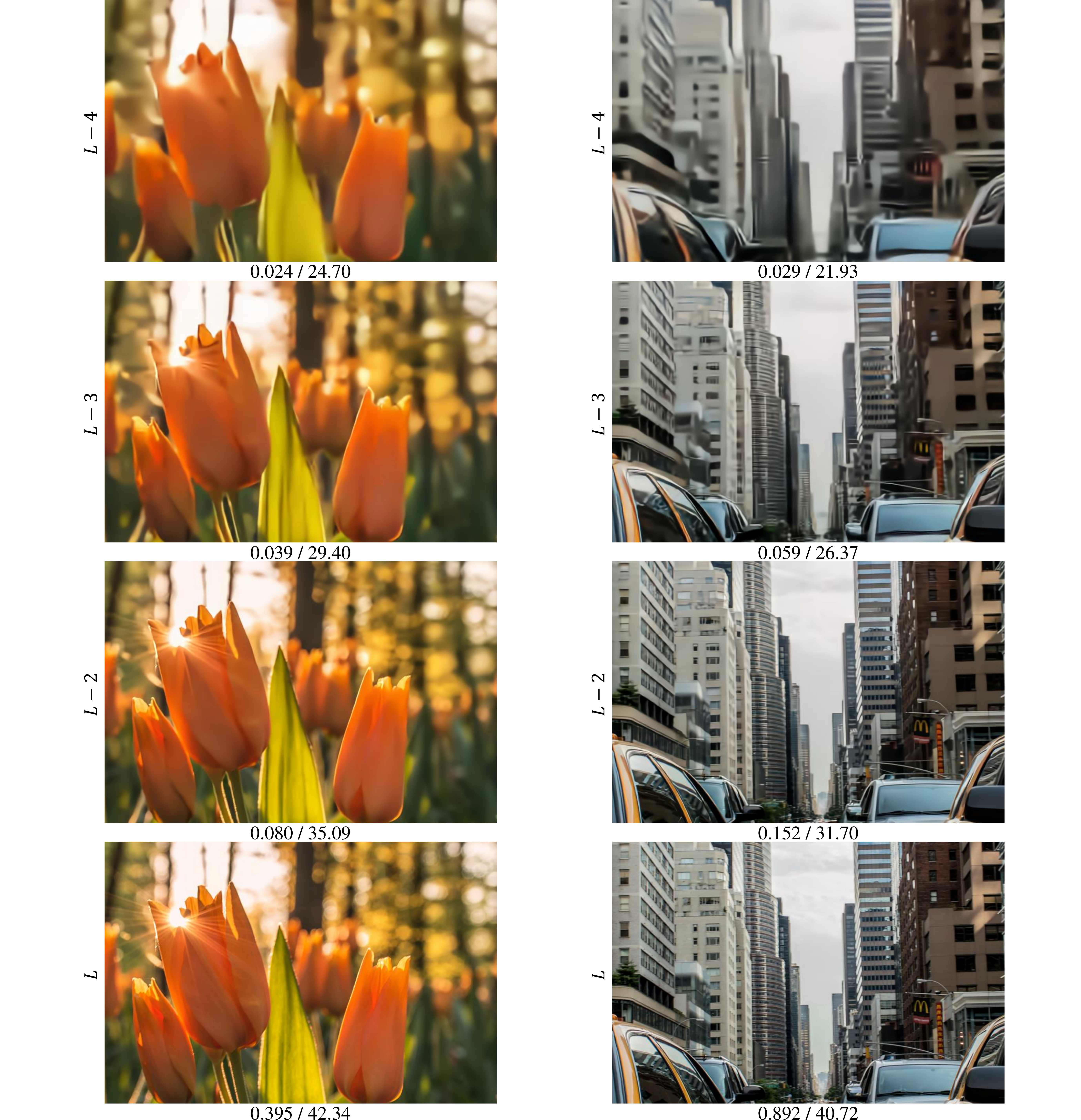}
    \end{center}
    \vspace*{-0.5cm}
    \caption
    {
        Reconstructed images of ``ales-krivec15949.png'' and ``andrew-ruiz-376.png.'' The bitrates (bpp) and PSNRs (dB) are reported below each image.
    }
    \label{fig:qualitative_clic_1}
\end{figure*}

\begin{figure*}[b]
    \begin{center}
    \includegraphics[width=\linewidth]{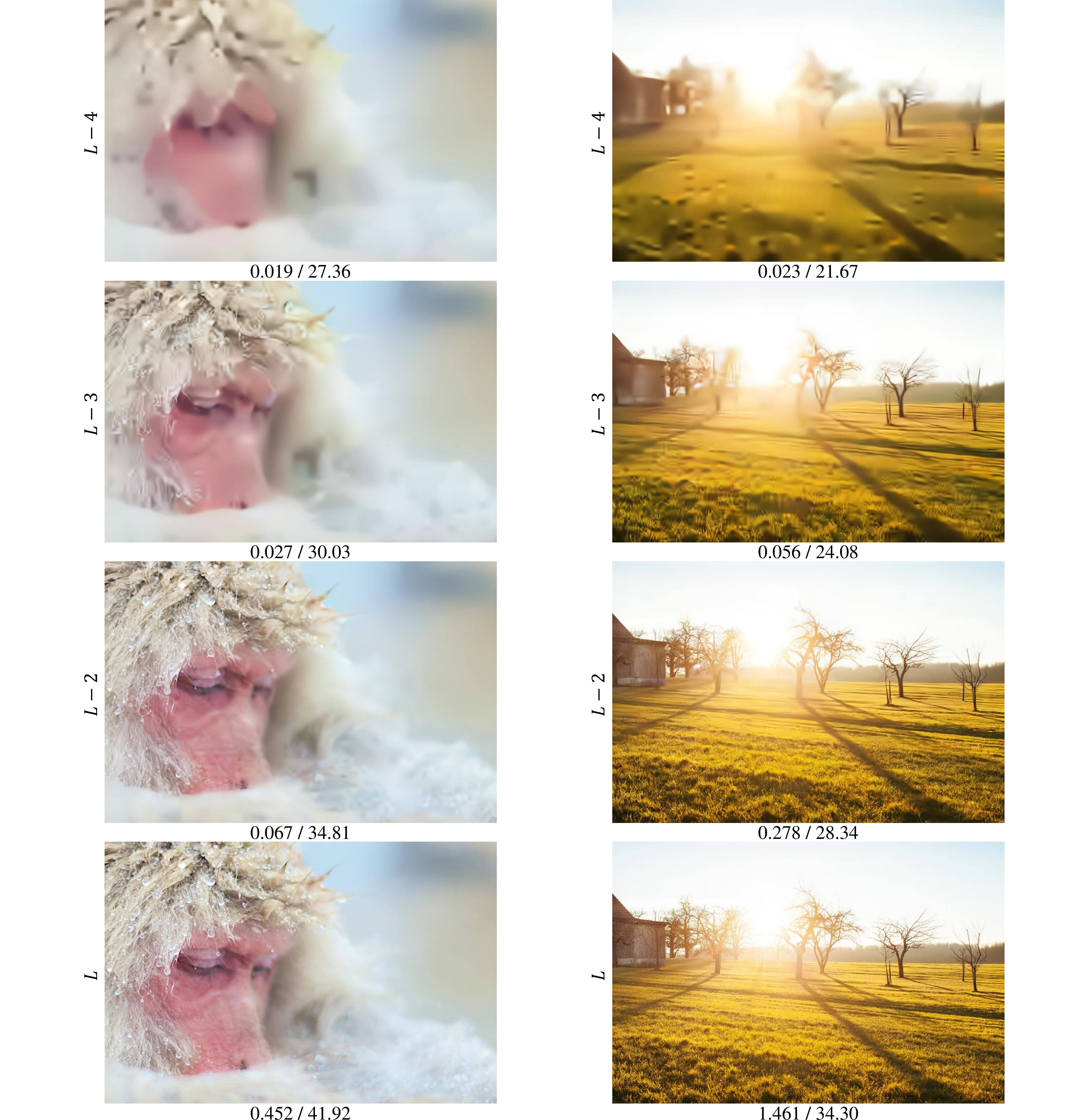}
    \end{center}
    \vspace*{-0.5cm}
    \caption
    {
        Reconstructed images of ``nomao-saeki-33553.png'' and ``philipp-reiner-207.png.'' The bitrates (bpp) and PSNRs (dB) are reported below each image.
    }
    \label{fig:qualitative_clic_2}
\end{figure*}

\begin{figure*}[b]
    \begin{center}
    \includegraphics[width=\linewidth]{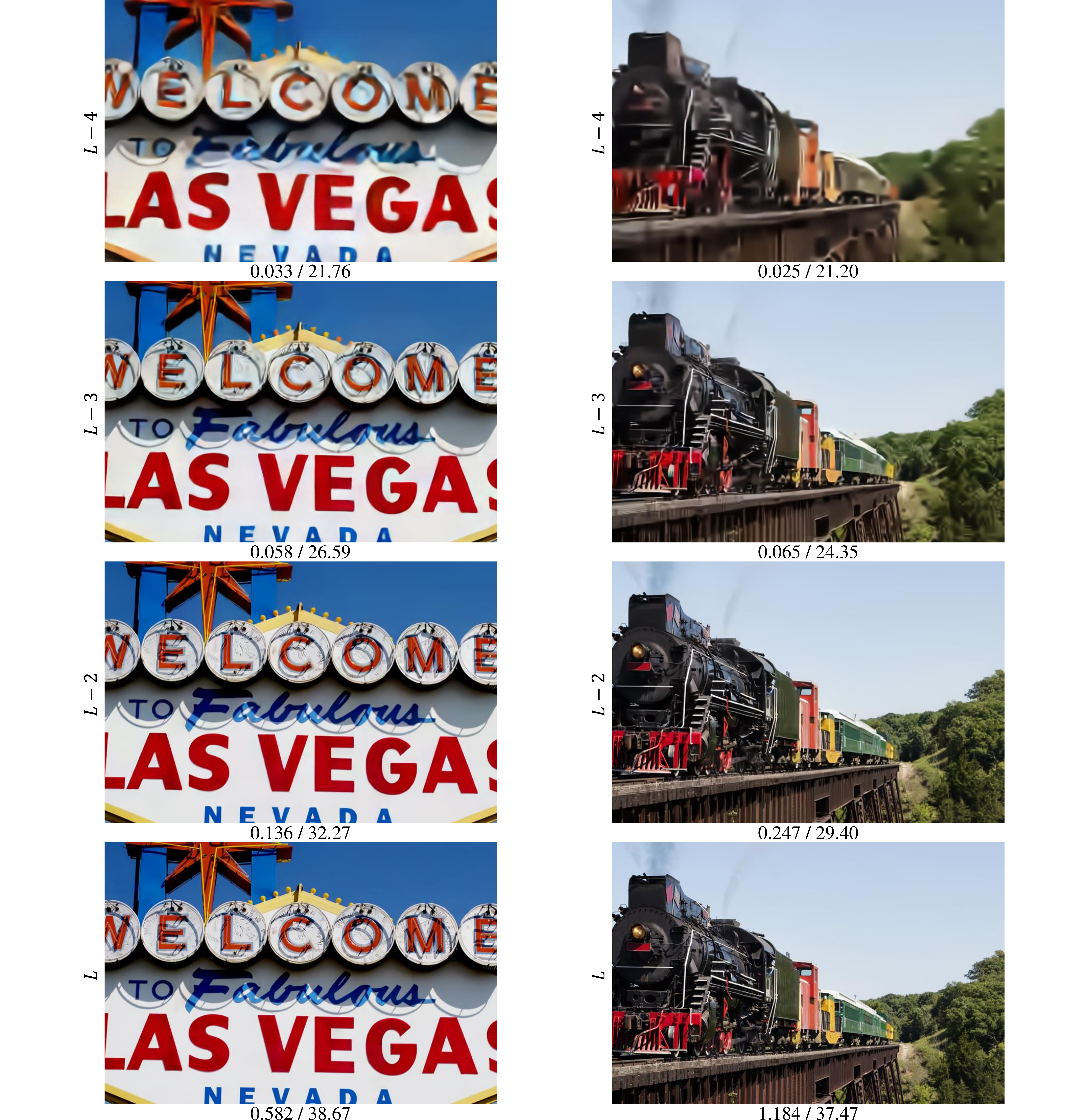}
    \end{center}
    \vspace*{-0.5cm}
    \caption
    {
        Reconstructed images of ``000505\_TE\_1336x872.png'' and ``000505\_TE\_1336x872.png.'' The bitrates (bpp) and PSNRs (dB) are reported below each image.
    }
    \label{fig:qualitative_jpegai_1}
\end{figure*}

\begin{figure*}[b]
    \begin{center}
    \includegraphics[width=\linewidth]{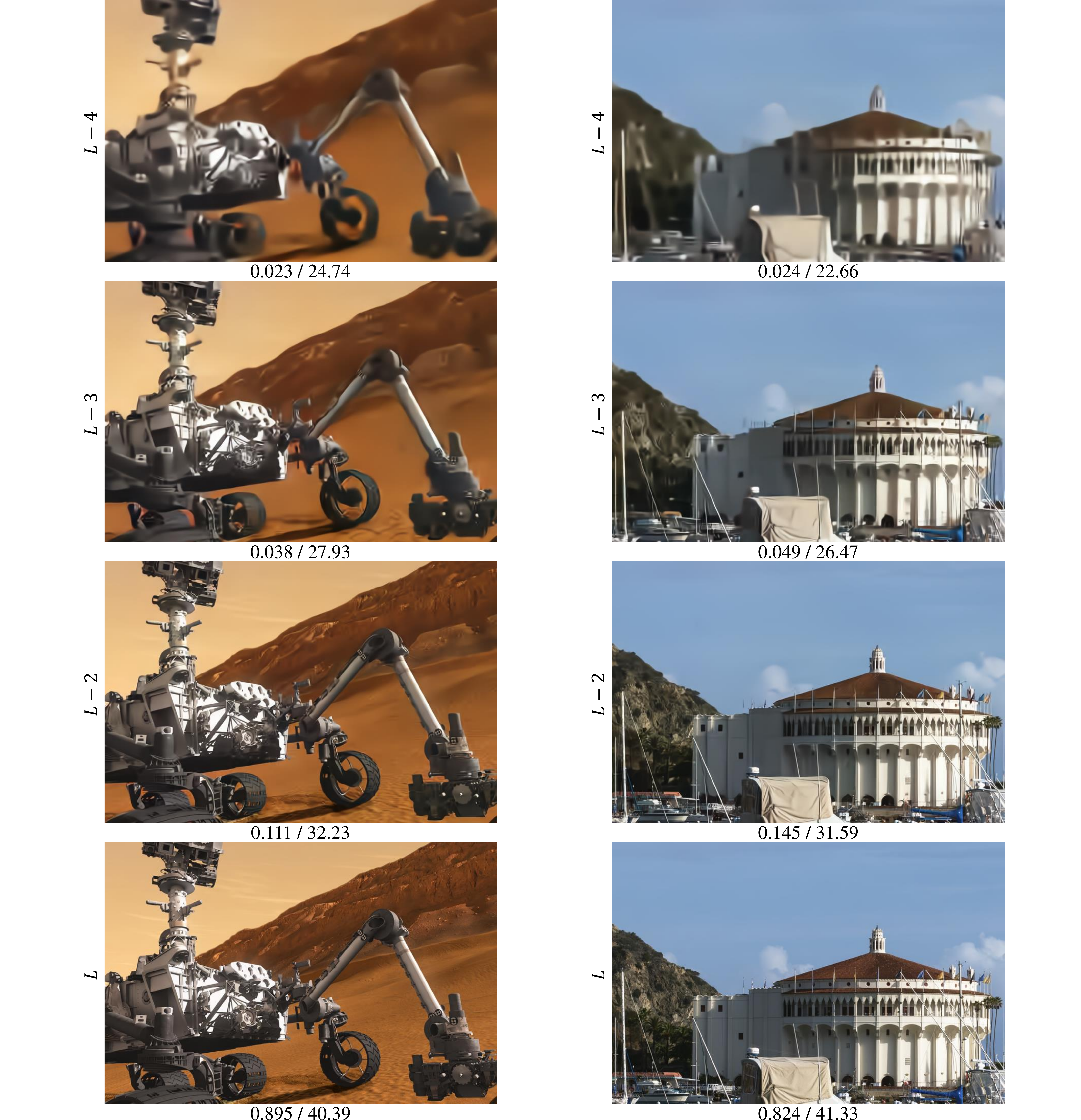}
    \end{center}
    \vspace*{-0.5cm}
    \caption
    {
        Reconstructed images of ``00010\_TE\_2000x1128.png'' and ``00015\_TE\_3680x2456.png.'' The bitrates (bpp) and PSNRs (dB) are reported below each image.
    }
    \label{fig:qualitative_jpegai_2}
\end{figure*}

\end{document}